\documentclass[submission,copyright]{eptcs}
\providecommand{\event}{QPL 2014} 

\usepackage{amsmath}
\usepackage{amsthm}
\usepackage{amssymb}
\usepackage{tikz}
\usetikzlibrary{calc}
\usepackage[arc,all]{xy}
\usepackage{graphicx}
\usepackage{wrapfig}

\usepackage[only,llbracket,rrbracket]{stmaryrd}

\let\Sect\S
\def\S{\mathit{S}}

\newcommand{\mathrotate}[2]{\text{\rotatebox[origin=c]{#1}{$#2$}}}

\newcommand{\cmp}{\mathrel{\circ}}
\newcommand{\rest}[2]{{#1}{\upharpoonright}{#2}}
\newcommand{\Scott}[1]{\llbracket{#1}\rrbracket}
\newcommand{\pb}[1][0,0]{%
	\begin{picture}(0,0)(#1)%
	\put(8,-8){\line(0,1){8}}%
	\put(8,-8){\line(-1,0){8}}%
	\end{picture}%
}
\newcommand{\pw}{\mathcal{P}}

\newcommand{\op}{\mathrm{op}}
\newcommand{\RRp}{{\mathbb{R}_{\geqslant 0}}}
\newcommand{\RR}{\mathbb{R}}
\newcommand{\BB}{\mathbb{B}}
\newcommand{\D}{\mathcal{D}}
\newcommand{\T}{T}
\newcommand{\supp}[1]{\textrm{supp}(#1)}

\newcommand{\rel}{\mathrel{\text{\ooalign{$\to$\crcr\hss\raisebox{0.08ex}{$\shortmid\mspace{1.25mu}$}\hss}}}}
\newcommand{\Sets}{\textbf{Sets}}
\newcommand{\Fib}{\textbf{Fib}}
\newcommand{\Map}{\textbf{-Map}}

\newcommand{\Rel}{\textbf{Rel}}

\newcommand{\RRel}{\textbf{-}\Rel}
\newcommand{\Stoch}{\textbf{Stoch}}
\newcommand{\Kl}{\mathrm{Kl}}
\newcommand{\cateq}{\simeq}
\newcommand{\nec}[1]{\mathop{[\mspace{1mu}#1\mspace{1mu}]}}
\newcommand{\pos}[1]{\mathop{\langle #1 \rangle}}
\newcommand{\Det}{\textrm{Det}}
\newcommand{\M}{M}

\newcommand{\NewTheorem}[2]{%
	\newtheorem{#1}{#2}
	\expandafter\def\csname #1autorefname\endcsname{#2}
}
\NewTheorem{theorem}{Theorem}
\NewTheorem{lemma}{Lemma}
\NewTheorem{proposition}{Proposition}
\theoremstyle{definition}
\NewTheorem{corollary}{Corollary}
\NewTheorem{definition}{Definition}
\NewTheorem{example}{Example}
\NewTheorem{fact}{Fact}
\NewTheorem{observation}{Observation}
\NewTheorem{remark}{Remark}

\title{%
Stochastic Relational Presheaves and\\
Dynamic Logic for Contextuality%
}
\author{
Kohei Kishida\thanks{%
Kishida's research has been supported by the grant FA9550-12-1-0136 of the U.S. AFOSR\@.
He is grateful to the audience at QPL 2014 as well as to the anonymous referees for their insightful comments and suggestions, which improved this paper.
Grateful acknowledgment also goes to the audience at the 2013 Workshop on Information and Processes, at which an earlier version of this paper was presented, and in particular to Prof.\ Michael Mislove and Tulane University, who hosted and funded the workshop.
}
\institute{Department of Computer Science\\
University of Oxford\\
Oxford, United Kingdom}
\email{kohei.kishida@cs.ox.ac.uk}
}
\def\titlerunning{Dynamic Logic for Contextuality}
\def\authorrunning{K. Kishida}
\begin{document}
\maketitle

\begin{abstract}
Presheaf models \cite[etc.]{cat97,win97}\ provide a formulation of labelled transition systems that is useful for, among other things, modelling concurrent computation.
This paper aims to extend such models further to represent stochastic dynamics such as shown in quantum systems.
After reviewing what presheaf models represent and what certain operations on them mean in terms of notions such as internal and external choices, composition of systems, and so on, I will show how to extend those models and ideas by combining them with ideas from other category-theoretic approaches to relational models \cite{her11} and to stochastic processes \cite[etc.]{fri09,bae11,jac11}.
It turns out that my extension yields a transitional formulation of sheaf-theoretic structures that Abramsky and Brandenburger \cite{abr11} proposed to characterize non-locality and contextuality.
An alternative characterization of contextuality will then be given in terms of a dynamic modal logic of the models I put forward.
\end{abstract}

\section{Introduction}

The goal of this paper is to devise a formalism of semantic structure for dynamic logic that is suitable for expressing stochastic dynamics such as shown in quantum systems.
Essential features of stochastic dynamics I aim to capture include
\begin{itemize}
\item
the distinction and interaction between \emph{internal} and \emph{external choices}, that is, non-deterministic branchings that are made within a system and that are made by external agents or experimenters;
\item
the distinction and interaction between what is \emph{globally} the case in an entire system and what is \emph{locally} the case in a subsystem.
\end{itemize}
In particular, the resulting semantics and logic shall be general enough to accommodate both the presence and the absence of (typically quantum) \emph{non-locality} and \emph{contextuality}, but at the same time expressive enough to provide logical characterization for non-locality and contextuality.

I achieve my goal by integrating three frameworks of categorical approaches that have been proposed to modelling non-deterministic and stochastic processes.
Firstly, my formalism will be based on
\begin{enumerate}
\item\label{item:presheaf}
Presheaves as labelled transition models for concurrency (Winskel \textit{et al}.\ \cite{cat97,win97}, etc.).
I show how the presheaf structure can be used to capture notions that are essential to my goal, such as internal and external choices, composition of multipartite systems, and so on.
\end{enumerate}
Then I extend this setting in two aspects, by admitting non-trees and by adding probabilities.
I attain these extensions by integrating the following ideas into my formalism.
\begin{enumerate}
\addtocounter{enumi}{1}
\item\label{item:rel}
Kripke relational semantics in terms of Kripke frames as functors from labels to the category $\Rel$ of sets and relations (Hermida \cite{her11}).
Integrating this idea with the presheaf framework admits presheaf-like models as transition systems of non-tree forms.
I will also lay out motivation for admitting non-trees.
(One mode of this integration has already been given in Soboci\'nski \cite{sob12};
yet the mode of integration I propose in this paper is different and not equivalent.)
\item\label{item:stoch}
The category of stochastic maps, or equivalently the Kleisli category of the distribution monad (the idea goes back at least to Lawvere \cite{law62}; it is also studied recently by Fritz \cite{fri09}, Baez \textit{et al}.\ \cite{bae11}, Fong \cite{fon12}, etc., in the former formulation, by Jacobs \cite{jac11}, etc., in the latter formulation).
How to add probabilities to presheaf models is a question posed in the concluding part of Varacca \cite{var03};
I answer this question by using structures closely related, though not equivalent, to stochastic maps.
\end{enumerate}
These extensions give semantic structures on which I define a dynamic and probabilistic logic.

To demonstrate that the resulting semantics and logic achieve the goal mentioned above, I will show how they capture non-locality and contextuality.
In particular, the semantics gives an alternative, transitional formulation to
%
%
a sheaf-theoretic approach to non-locality and contextuality (Abramsky and Brandenburger \cite{abr11}, etc.).
This approach provides a sheaf-theoretic expression for, among other things, measurement scenarios in quantum mechanics, and characterizes non-locality and contextuality found in such scenarios in terms of non-existence of global sections.
The transitional formulation I give to this approach leads to an alternative, dynamic-logical characterization of non-locality and contextuality.

\section{Presheaf Models for Measurements}\label{sec:presheaf}

This section reviews presheaves over trees as labelled transition systems (see \cite{cat97,win97}).
Rather than giving new definitions or theorems, this section is concerned with conceptually laying out how to use the familiar notions of presheaf and fibration to represent features of non-deterministic processes that are essential to the goal of this paper.

\subsection{Trees and Presheaves of Non-Deterministic Choices}

Here I lay out the key idea of how to use a presheaf over a tree as a labelled transition system, or LTS for short, in a manner suitable for representing different kinds of non-determinacy in stochastic processes.

As in the standard terminology, by a ``measurement scenario of $(n, k, \ell)$-type'' let us mean a Bell-type scenario of (typically quantum) measurements that involves $n$ parts (or experimenters), each of which (or whom) chooses one from $k$ measurements, each of which has $\ell$ outcomes.
For instance, in a $(1, 2, 2)$ scenario, Alice chooses one from two measurements, $a$ and $a'$, each of which has two outcomes, $0$ and $1$.
This simple scenario can be represented by the following tree $L$ and presheaf $\S$ over $L$.
%
%
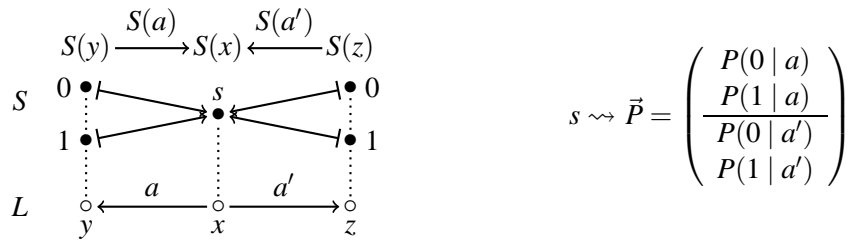
\begin{align*}
\begin{tikzpicture}[x=40pt,y=40pt,thick,label distance=-0.25em,baseline=(current bounding box.center)]
\node (l) at (-1.25,0) {};
\node (r) at (1.25,0) {};
\node (d) at (0,-0.25) {};
\node (u) at (0,0.25) {};
\node [inner sep=0.1em,label=below:{$x$}] (R) at (0,0) {$\circ$};
\node [inner sep=0.1em,label=below:{$y$}] (Ra) at ($ (R) + (l) $) {$\circ$};
\node [inner sep=0.1em,label=below:{$z$}] (Ra') at ($ (R) + (r) $) {$\circ$};
\node [inner sep=0.1em,label=above:{$s$}] (O) at ($ (R) + (0,0.875) $) {$\bullet$};
\node [inner sep=0.1em,label=left:{$0$}] (Oa0) at ($ (O) + (l) + (u) $) {$\bullet$};
\node [inner sep=0.1em,label=left:{$1$}] (Oa1) at ($ (O) + (l) + (d) $) {$\bullet$};
\node [inner sep=0.1em,label=right:{$0$}] (Oa'0) at ($ (O) + (r) + (u) $) {$\bullet$};
\node [inner sep=0.1em,label=right:{$1$}] (Oa'1) at ($ (O) + (r) + (d) $) {$\bullet$};
\draw [->] (R) -- (Ra) node [pos=0.5,above] {$a$};
\draw [->] (R) -- (Ra') node [pos=0.5,above] {$a'$};
\draw [<-|] (O) -- (Oa0);
\draw [<-|] (O) -- (Oa1);
\draw [<-|] (O) -- (Oa'0);
\draw [<-|] (O) -- (Oa'1);
\draw [dotted] (R) -- (O);
\draw [dotted] (Ra) -- (Oa0);
\draw [dotted] (Ra') -- (Oa'0);
\node [inner sep=0.1em] (L) at (-1.875,0) {$L$};
\node [inner sep=0.1em] (S) at (-1.875,1) {$\S$};
\node [inner sep=0.1em] (Sx) at ($ (O) + (0,0.625) $) {$\S(x)$};
\node [inner sep=0.1em] (Sy) at ($ (Sx) + (l) $) {$\S(y)$};
\node [inner sep=0.1em] (Sz) at ($ (Sx) + (r) $) {$\S(z)$};
\draw [->] (Sy) -- (Sx) node [pos=0.5,above] {$\S(a)$};
\draw [->] (Sz) -- (Sx) node [pos=0.5,above] {$\S(a')$};
\end{tikzpicture}
&&
s \rightsquigarrow \, \vec{P} = \left(
\begin{array}{c}
P(0 \mid a) \\
P(1 \mid a) \\ \hline
P(0 \mid a') \\
P(1 \mid a')
\end{array}
\right)
\end{align*}
The binary branching in $L$ represents the choice Alice makes outside the system, choosing from two measurements $a$ and $a'$.
Then regard $\S$ as a transition system, reading ``$\mapsto$'' backward as transition ``$\leftarrow$'';
each such edge of transition in $\S$ is labelled with an edge in $L$---%
for instance, those in $\S(a)$ above are labelled with $a$, representing possible outcomes the system has for Alice's choice of $a$.
So the binary branching in $\S(a)$ represents the system having two outcomes for measurement $a$.
One of our objectives is to assign probabilities to such branchings, so that, in the picture above for instance, the state $s$ can be (at least partially) specified by the vector of probabilities $\vec{P}$ to the right of the picture above.

Note that the representation just given involves two kinds of choice.
Put in general terms,
when we describe a system and agents external to the system,
\begin{itemize}
\item
The agents may be able to choose from different ways to interfere or interact with the system.
We call these choices \emph{external choices}, and represent them with branching in the base tree.
\item
The system may behave by itself non-deterministically---%
sometimes in response to external choices, but sometimes simply as time passes---%
with several possible outcomes.
We call these choices \emph{internal choices}, and represent them with branching in function components of the presheaf.
\end{itemize}
In short, external choice resides in the base tree $L$;
internal choice resides in (function components of) the presheaf $\S$.
This is the slogan for our use of presheaves $\S$ over trees $L$ as $L$-labelled transition systems.

In fact, not just the distinction between internal and external choices, the presheaf structure also gives us several useful ways to control descriptions of these choices---%
for instance, to shift the boundary between the internal and external.
We will see this in \autoref{sec:control.description}.
Before doing so, it is useful to observe that the presheaves over a tree are equivalent to the \emph{fibrations} over the tree (which should be quite obvious from the picture above).
Let us recall

\begin{definition}
A bundle (i.e., monotone map) $\pi : \S \to L$ of posets is called a \emph{fibration} (over $L$) if, whenever $x \leqslant_L \pi(t)$, there is a unique $s \in \pi^{-1}(x)$ such that $s \leqslant_\S t$.
Write $\Fib$ for the category of posets and fibrations.
\end{definition}

We should note that if $\pi : \S \to L$ is a fibration and $L$ is a tree then $\S$ is also a tree.
Then it is easy to show the following (we provide a proof rather as a review of notation).

\begin{fact}\label{thm:presheaf.fibration.equiv}
$\Sets^{L^\op} \cateq \Fib / L$ for any poset $L$.
\end{fact}

\begin{proof}
A presheaf $\S : L^\op \to \Sets$ yields a fibration with the projection $\pi : \S \to L$ from the dependent sum
%
%
$\S := \sum_{x \in L} \S(x) = \{\, (x, s) \mid x \in L \text{ and } s \in \S(x) \,\}$
to $L$ and the order $\leqslant_\S$ on $\S$ such that $(x, s) \leqslant_\S (y, t)$ iff $x \leqslant_L y$ and $s = \S(x, y)(t)$.
A fibration $\pi : \S \to L$ yields a presheaf $\S : L^\op \to \Sets$ by letting $\S(x) = \pi^{-1}(x)$ for $x \in L$ and, whenever $x \leqslant_L y$, defining $\S(x, y) : \pi^{-1}(y) \to \pi^{-1}(x)$ so that $\S(x, y)(t)$ for $t \in \pi^{-1}(y)$ is the unique $s \in \pi^{-1}(x)$ such that $s \leqslant_\S t$.

Given presheaves $\S, T : L^\op \to \Sets$ and corresponding fibrations $\pi_\S : \S \to L$, $\pi_T : T \to L$, the natural transformations from $\S$ to $T$ are just the monotone maps $f : \S \to T$ over $L$ (meaning $\pi_T \cmp f = \pi_\S$), but any such monotone map $f$ can easily be shown to be a fibration.
\end{proof}

We will invoke this presheaf-fibration equivalence extensively in the rest of this paper.

\subsection{Controlling System Descriptions}\label{sec:control.description}

Given presheaf-fibration descriptions of non-deterministic processes with internal and external choices, we can take further advantage of operations on the presheaf-fibration structure to control the descriptions.

A family of operations that will later prove useful is done by change of base.
One such operation is to precompose a given presheaf $\S : {L_1}^\op \to \Sets$ with an embedding $m : L_0 \rightarrowtail L_1$, obtaining a new presheaf $\S \cmp m^\op : {L_0}^\op \to \Sets$.
Since some points, or ``stages'', of $L_1$ are ``omitted'' in $L_0$, the precomposition makes the model ``forget'' what takes place at these omitted stages.
For instance, take $m : L_0 \rightarrowtail L_1$ as on the left of \eqref{eq:change.of.base} below, and let $a$ and $b$ represent measurements by Alice and by Bob.
Then a presheaf $\S : {L_1}^\op \to \Sets$ carries information as to the original states (in $\S(x)$), the possible outcomes of $a$ (in $\S(y)$), and then the possible further outcomes of $b$ (in $\S(z)$).
In contrast, the presheaf $\S \cmp m^\op : {L_0}^\op \to \Sets$ carries the same information as to the original states (in $\S(x)$) and the outcomes of both measurements (in $\S(z)$), but it has no information as to the process in between (or, indeed, even as to whether $a$ is performed before, after, or at the same time as $b$).
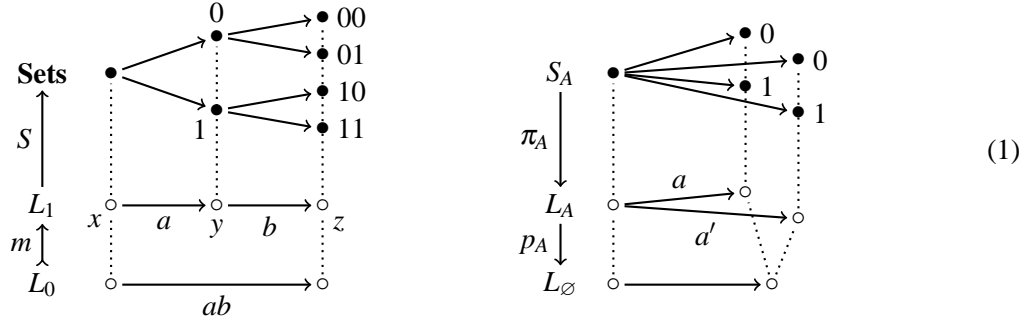
\begin{align}
\label{eq:change.of.base}
\begin{tikzpicture}[x=40pt,y=40pt,thick,label distance=-0.25em,baseline=(current bounding box.center)]
\node (i) at (1,0) {};
\node (j) at (1.5,0) {};
\node (l) at (-0.25,0.125) {};
\node (r) at (0.25,-0.125) {};
\node (u) at (0,0.175) {};
\node (d) at (0,-0.175) {};
\node [inner sep=0.1em,label={[label distance=-0.5em]235:{$x$}}] (KAO) at (0,0) {$\circ$};
\node [inner sep=0.1em,label=below:{$y$}] (KAOa) at ($ (KAO) + (i) $) {$\circ$};
\node [inner sep=0.1em,label={[label distance=-0.5em]315:{$z$}}] (KAOab) at ($ (KAOa) + (i) $) {$\circ$};
\draw [->] (KAO) -- (KAOa) node [pos=0.5,below] {$a$};
\draw [->] (KAOa) -- (KAOab) node [pos=0.5,below] {$b$};
\node [inner sep=0.1em] (TAO) at ($ (KAO) + (0,1.25) $) {$\bullet$};
\node [inner sep=0.1em,label=above:{$0$}] (TAOa0) at ($ (TAO) + (i) + (u) + (u) $) {$\bullet$};
\node [inner sep=0.1em,label={[label distance=-0.5em]235:{$1$}}] (TAOa1) at ($ (TAO) + (i) + (d) + (d) $) {$\bullet$};
\node [inner sep=0.1em,label=right:{$00$}] (TAOa0b0) at ($ (TAOa0) + (i) + (u) $) {$\bullet$};
\node [inner sep=0.1em,label=right:{$01$}] (TAOa0b1) at ($ (TAOa0) + (i) + (d) $) {$\bullet$};
\node [inner sep=0.1em,label=right:{$10$}] (TAOa1b0) at ($ (TAOa1) + (i) + (u) $) {$\bullet$};
\node [inner sep=0.1em,label=right:{$11$}] (TAOa1b1) at ($ (TAOa1) + (i) + (d) $) {$\bullet$};
\draw [->] (TAO) -- (TAOa0);
\draw [->] (TAO) -- (TAOa1);
\draw [->] (TAOa0) -- (TAOa0b0);
\draw [->] (TAOa0) -- (TAOa0b1);
\draw [->] (TAOa1) -- (TAOa1b0);
\draw [->] (TAOa1) -- (TAOa1b1);
\draw [dotted] (KAOa) -- (TAOa0);
\node [inner sep=0.2em] (KA) at ($ (KAO) + (-0.65,0) $) {$L_1$};
\node [inner sep=0.2em] (TA) at ($ (KA) + (0,1.25) $) {$\Sets$};
\draw [<-] (TA) -- (KA) node [pos=0.5,left] {$\S$};
\node [inner sep=0.1em] (K0O) at ($ (KAO) + (0,-0.75) $) {$\circ$};
\node [inner sep=0.1em] (K0Oi) at ($ (K0O) + (i) + (i) $) {$\circ$};
\draw [->] (K0O) -- (K0Oi) node [pos=0.5,below] {$ab$};
\draw [dotted] (K0O) -- (KAO) -- (TAO);
\draw [dotted] (K0Oi) -- (KAOab) -- (TAOa0b0);
\node [inner sep=0.2em] (K0A) at ($ (KA) + (0,-0.75) $) {$L_0$};
\draw [<-<] (KA) -- (K0A) node [pos=0.5,left] {$m$};
\node (uu) at (0,0.25) {};
\node (dd) at (0,-0.25) {};
\node [inner sep=0.1em] (LAO) at (4.75,0) {$\circ$};
\node [inner sep=0.1em] (LAOa) at ($ (LAO) + (j) + (l) $) {$\circ$};
\node [inner sep=0.1em] (LAOa') at ($ (LAO) + (j) + (r) $) {$\circ$};
\draw [->] (LAO) -- (LAOa) node [pos=0.5,above] {$a$};
\draw [->] (LAO) -- (LAOa') node [pos=0.5,below] {$a'$};
\node [inner sep=0.1em] (SAO) at ($ (LAO) + (0,1.25) $) {$\bullet$};
\node [inner sep=0.1em,label=right:{$0$}] (SAOa0) at ($ (SAO) + (j) + (l) + (uu) $) {$\bullet$};
\node [inner sep=0.1em,label=right:{$1$}] (SAOa1) at ($ (SAO) + (j) + (l) + (dd) $) {$\bullet$};
\node [inner sep=0.1em,label=right:{$0$}] (SAOa'0) at ($ (SAO) + (j) + (r) + (uu) $) {$\bullet$};
\node [inner sep=0.1em,label=right:{$1$}] (SAOa'1) at ($ (SAO) + (j) + (r) + (dd) $) {$\bullet$};
\draw [->] (SAO) -- (SAOa0);
\draw [->] (SAO) -- (SAOa1);
\draw [->] (SAO) -- (SAOa'0);
\draw [->] (SAO) -- (SAOa'1);
\node [inner sep=0.2em] (LA) at ($ (LAO) + (-0.5,0) $) {$L_A$};
\node [inner sep=0.2em] (SA) at ($ (LA) + (0,1.25) $) {$\S_A$};
\draw [->] (SA) -- (LA) node [pos=0.5,left] {$\pi_A$};
\node [inner sep=0.1em] (L0O) at ($ (LAO) + (0,-0.75) $) {$\circ$};
\node [inner sep=0.1em] (L0Oi) at ($ (L0O) + (j) $) {$\circ$};
\draw [->] (L0O) -- (L0Oi);
\draw [dotted] (L0O) -- (LAO) -- (SAO);
\draw [dotted] (L0Oi) -- (LAOa) -- (SAOa0);
\draw [dotted] (L0Oi) -- (LAOa') -- (SAOa'0);
\node [inner sep=0.2em] (L0A) at ($ (LA) + (0,-0.75) $) {$L_\varnothing$};
\draw [->] (LA) -- (L0A) node [pos=0.5,left] {$p_A$};
\end{tikzpicture}
\end{align}

Another is to compose fibrations $\pi : \S \to L_0$ and $p : L_0 \to L_1$, obtaining a new fibration $p \cmp \pi : \S \to L_1$.
In $\pi$, branchings in $L_0$ represent external choices, but some of them are internal choices in $p$;
so the composition ``internalizes'' these external choices.
Take $\pi_A$ and $p_A$ as on the right of \eqref{eq:change.of.base} above.
$\pi_A$ describes Alice as an agent external to a system who externally chooses from measurements $a$ and $a'$.
On the other hand, $p_A \cmp \pi_A$ describes a bigger system encompassing Alice---%
so that we simply watch the bigger system internally choose from the four outcomes, ``Alice performs $a$ and gets outcome $0$'', etc.

In fact, such composition of fibrations can be used to compose descriptions of several systems into a description of a multipartite system.
The fibration $\pi_A$ in the picture above describes a $(1, 2, 2)$-scenario for Alice.
Take an isomorphic $\pi_B : \S_B \to L_B$ to describe a $(1, 2, 2)$-scenario for Bob.
Then a fibration $\pi_{AB} : \S_{AB} \to L_{AB}$ for the composed $(2, 2, 2)$-scenario is obtained as follows:
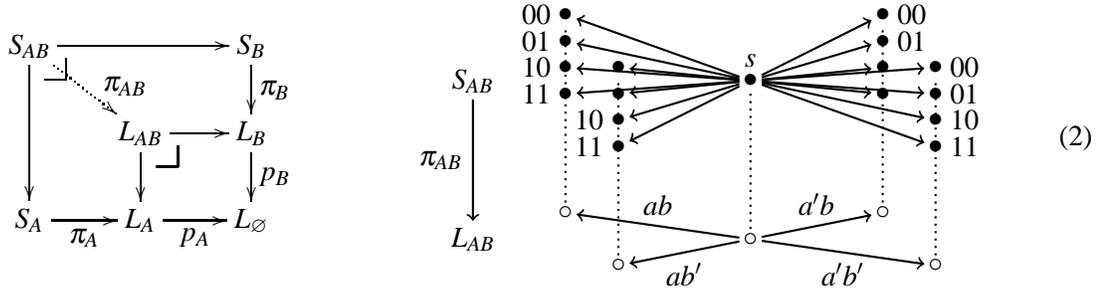
\begin{align}
\label{eq:multipartite}
\text{
\begin{tabular}[c]{@{}c@{}}
\xymatrix@!0@C+1.5pc@R+0.75pc{
\S_{AB}
\pb[2,2]
\ar[2,0]
\ar[0,2]
\ar@{.>}[1,1]^(.65)*{\pi_{AB}}
&&
\S_B
\ar[1,0]^-*{\pi_B}
\\
&
L_{AB}
\pb[2,2]
\ar[1,0]
\ar[0,1]
&
L_B
\ar[1,0]^-*{p_B}
\\
\S_A
\ar[0,1]_-*{\pi_A}
&
L_A
\ar[0,1]_-*{p_A}
&
L_\varnothing
}
\end{tabular}
}
&&
\begin{tikzpicture}[x=40pt,y=40pt,thick,label distance=-0.25em,baseline=(current bounding box.center)]
\node [inner sep=0.1em] (R) at (0,0) {$\circ$};
\node [inner sep=0.1em] (Rab) at (-1.75,0.25) {$\circ$};
\node [inner sep=0.1em] (Rab') at (-1.25,-0.25) {$\circ$};
\node [inner sep=0.1em] (Ra'b) at (1.25,0.25) {$\circ$};
\node [inner sep=0.1em] (Ra'b') at (1.75,-0.25) {$\circ$};
\node [inner sep=0.1em,label=above:{$s$}] (O) at ($ (R) + (0,1.5) $) {$\bullet$};
\node [inner sep=0.1em] (Oab) at ($ (Rab) + (0,1.5) $) {};
\node [inner sep=0.1em,label=left:{$00
	$}] (Oa0b0) at ($ (Oab) + (0,0.375) $) {$\bullet$};
\node [inner sep=0.1em,label=left:{$01
	$}] (Oa0b1) at ($ (Oab) + (0,0.125) $) {$\bullet$};
\node [inner sep=0.1em,label=left:{$10$}] (Oa1b0) at ($ (Oab) + (0,-0.125) $) {$\bullet$};
\node [inner sep=0.1em,label=left:{$11$}] (Oa1b1) at ($ (Oab) + (0,-0.375) $) {$\bullet$};
\node [inner sep=0.1em] (Oab') at ($ (Rab') + (0,1.5) $) {};
\node [inner sep=0.1em] (Oa0b'0) at ($ (Oab') + (0,0.375) $) {$\bullet$};
\node [inner sep=0.1em] (Oa0b'1) at ($ (Oab') + (0,0.125) $) {$\bullet$};
\node [inner sep=0.1em,label=left:{$10
	$}] (Oa1b'0) at ($ (Oab') + (0,-0.125) $) {$\bullet$};
\node [inner sep=0.1em,label=left:{$11
	$}] (Oa1b'1) at ($ (Oab') + (0,-0.375) $) {$\bullet$};
\node [inner sep=0.1em] (Oa'b) at ($ (Ra'b) + (0,1.5) $) {};
\node [inner sep=0.1em,label=right:{$00
	$}] (Oa'0b0) at ($ (Oa'b) + (0,0.375) $) {$\bullet$};
\node [inner sep=0.1em,label=right:{$01
	$}] (Oa'0b1) at ($ (Oa'b) + (0,0.125) $) {$\bullet$};
\node [inner sep=0.1em] (Oa'1b0) at ($ (Oa'b) + (0,-0.125) $) {$\bullet$};
\node [inner sep=0.1em] (Oa'1b1) at ($ (Oa'b) + (0,-0.375) $) {$\bullet$};
\node [inner sep=0.1em] (Oa'b') at ($ (Ra'b') + (0,1.5) $) {};
\node [inner sep=0.1em,label=right:{$00$}] (Oa'0b'0) at ($ (Oa'b') + (0,0.375) $) {$\bullet$};
\node [inner sep=0.1em,label=right:{$01$}] (Oa'0b'1) at ($ (Oa'b') + (0,0.125) $) {$\bullet$};
\node [inner sep=0.1em,label=right:{$10
	$}] (Oa'1b'0) at ($ (Oa'b') + (0,-0.125) $) {$\bullet$};
\node [inner sep=0.1em,label=right:{$11
	$}] (Oa'1b'1) at ($ (Oa'b') + (0,-0.375) $) {$\bullet$};
\draw [->] (R) -- (Rab) node [pos=0.5,above] {$ab$};
\draw [->] (R) -- (Rab') node [pos=0.5,below] {$ab'$};
\draw [->] (R) -- (Ra'b) node [pos=0.5,above] {$a'b$};
\draw [->] (R) -- (Ra'b') node [pos=0.5,below] {$a'b'$};
\draw [->] (O) -- (Oa0b0);
\draw [->] (O) -- (Oa0b1);
\draw [->] (O) -- (Oa1b0);
\draw [->] (O) -- (Oa1b1);
\draw [->] (O) -- (Oa0b'0);
\draw [->] (O) -- (Oa0b'1);
\draw [->] (O) -- (Oa1b'0);
\draw [->] (O) -- (Oa1b'1);
\draw [->] (O) -- (Oa'0b0);
\draw [->] (O) -- (Oa'0b1);
\draw [->] (O) -- (Oa'1b0);
\draw [->] (O) -- (Oa'1b1);
\draw [->] (O) -- (Oa'0b'0);
\draw [->] (O) -- (Oa'0b'1);
\draw [->] (O) -- (Oa'1b'0);
\draw [->] (O) -- (Oa'1b'1);
\draw [dotted] (R) -- (O);
\draw [dotted] (Rab) -- (Oa0b0);
\draw [dotted] (Rab') -- (Oa0b'0);
\draw [dotted] (Ra'b) -- (Oa'0b0);
\draw [dotted] (Ra'b') -- (Oa'0b'0);
\node [inner sep=0.2em] (L) at (-2.625,0) {$L_{AB}$};
\node [inner sep=0.2em] (S) at (-2.625,1.5) {$\S_{AB}$};
\draw [->] (S) -- (L) node [pos=0.5,left] {$\pi_{AB}$};
\end{tikzpicture}
\end{align}
That is, $\pi_{AB} = \pi_A \times_{L_\varnothing} \pi_B : \S_A \times_{L_\varnothing} \S_B \to L_A \times_{L_\varnothing} L_B$.
Put more conceptually, we use $L_\varnothing$ as a clock for synchronizing events in Alice's scenario and ones in Bob's, and then take simultaneous pairs of events from Alice's and Bob's scenarios.
We should note that the pair of projections from $\S_{AB}$ and $L_{AB}$ to $\S_A$ and $L_A$ represents the restriction of a description of what is globally the case in the bipartite system to a description of what is locally the case in Alice's subsystem---%
this is a tool crucial for the purpose of this paper, of capturing non-locality and contextuality.
We will see, for instance, that this projection has a role in characterizing the no-signalling property in fibrational terms in \autoref{sec:probability.presheaf}.

It may need stressing that $\S_{AB}$ described above is just a cartesian product (taken fiberwise over $L_\varnothing$)---%
rather than anything similar to a tensor product---%
of $\S_A$ and $\S_B$;
hence it does not by itself express any correlation between Alice's and Bob's measurement outcomes.
It is rather a transition-system expression for the $4 \times 4$ entries in a probability table describing a $(2, 2, 2)$-scenario.
Any correlation will be expressed by assigning probabilities to transitions in $\S_{AB}$;
we will see how in \autoref{sec:probability.presheaf}.

\section{Adding Probabilities to Presheaves}\label{sec:probability.presheaf}

This short section lays out how to add probabilities to the presheaf representation of non-deterministic processes given in \autoref{sec:presheaf}.
The definitions provided here will later be generalized in \autoref{sec:probability.relational.presheaf}, after a generalization of the presheaf representation is proposed in \autoref{sec:relational.presheaf}.

\subsection{Stochastic Presheaves}\label{sec:stochastic.presheaf}

Recall that in a description of a non-deterministic process with a presheaf $\S : L^\op \to \Sets$, for any edge $e = (x, y)$ of $L$ and state $s \in \S(x)$, the inverse image $\S(e)^{-1}(s) \subseteq \S(y)$ is the set of states to which the system may internally choose to transition from $s$ when $e$ is externally chosen.
Now we want to give probability to such an internal choice;
so let us achieve just that, with the following series of definitions.
They use the notion of $R$-distribution for a commutative semiring $R$;
see \cite[\Sect 2.3]{abr11} for its definition.
In particular, throughout this paper all distributions are assumed to be normalized and with finite support.

\begin{definition}
Fix a commutative semiring $R$.
Given any sets $X$ and $Y$, we define an \emph{$R$-map} from $X$ to $Y$ as any surjection $f : Y \twoheadrightarrow X$ (note the opposite direction) equipped with, for each $s \in X$, an $R$-distribution on $f^{-1}(s) \subseteq Y$, written $d^f_s$.
(We say that an $R$-map is \emph{on} its underlying surjection.)
\end{definition}

Obviously, we can achieve what we wanted above with an $\RRp$-map $f$ on $\S(e) : \S(y) \to \S(x)$ (assuming $\S(e)$ is surjective):
The distribution $d^f_s$ assigns to each $t \in \S(e)^{-1}(s)$ the probability $d^f_s(t)$ with which the system transitions from $s$ to $t$ (when $e$ is chosen).
To do this for the entire presheaf, we give

\begin{definition}\label{def:R.map.cat}
Given any two $R$-maps $f^R$ on $f : Z \twoheadrightarrow Y$ and $g^R$ on $g : Y \twoheadrightarrow X$, let their composition $f^R \cmp g^R$ be on $g \cmp f : Z \twoheadrightarrow X$ with, for each $s \in X$, an $R$-distribution $d^{g \cmp f}_s$ on $f^{-1}(g^{-1}(s)) \subseteq Z$ such that
\begin{gather}\label{eq:R.map.composition}
d^{g \cmp f}_s(u) = d^g_s(f(u)) \cdot d^f_{f(u)}(u) .
\end{gather}
Write $R\Map$ for the category of sets and $R$-maps.
(Clearly, the unique $R$-map on the identity map $1_X : X \to X$ is the identity on $X$ in $R\Map$.)
\end{definition}

The point of \eqref{eq:R.map.composition} should be clear:
When $s = g(t)$ and $t = f(u)$, the system transitions from $s$ to $t$ with probability $r = d^g_s(t)$ and from $t$ to $u$ with probability $r' = d^f_t(u)$;
so it transitions from $s$ to $u$ with probability $r \cdot r' = d^g_s(t) \cdot d^f_t(u) = d^{g \cmp f}_s(u)$.
(Note that the system can go from $s$ to $u$ through at most one $t$, since $f$ is a function.)
Then, finally,

\begin{definition}
An \emph{$R$-presheaf} over a category $\mathbf{C}$ is a contravariant functor from $\mathbf{C}$ to $R\Map$.
(We say that an $R$-presheaf is \emph{on} its underlying presheaf.)
\end{definition}

So, given a presheaf $\S : L^\op \to \Sets$ over a tree $L$ as an $L$-LTS, we assign probabilities to the internal choices in $\S$ by simply taking an $R$-presheaf on $\S$.

The presheaf-fibration equivalence (\autoref{thm:presheaf.fibration.equiv}) partially extends to $R$-presheaves:
We can define ``$R$-fibrations'' and prove that the equivalence extends to an essentially surjective and full functor from the category of rooted $R$-presheaves over a rooted tree $L$ to that of rooted $R$-fibrations over $L$ (we however omit the definitions and proof in this abstract).
This extended version is limited and no longer an equivalence, but good enough for practical purposes.
The core idea is that, given an $R$-presheaf $\S : L^\op \to \Rel$ that has a root $s \in \S(x)$, the ``horizontal'' assignment of probabilities $d^{\S(x, y)}_s(t)$ to all states $t \in \S$ can be turned into a ``vertical'' assignment of probabilities $d^{\pi}_y(t)$ on the fibration $\pi : \S \to L$ that corresponds to the underlying presheaf of $\S$.

Lastly, note that, although it may be proper to reserve the term ``probability'' to values of $\RRp$-distributions, in this paper I apply the term broadly to values of $R$-distributions in general.
Other interesting cases of $R$ include $\BB$, the booleans, and $\RR$, all the reals, both of which are discussed in \cite{abr11}.

\subsection{Example: No-Signalling}\label{sec:no.signalling}

Let us say that a commutative semiring $R$ is ``normalizable'' if, for every family $\{ r_i \}_{i \in I}$ of elements of $R$ such that $c := \sum_{i \in I} r_i \neq 0$, there is a family $\{ r'_i \}_{i \in I}$ of elements of $R$ such that $r'_j \cdot c = r_j$ for each $j \in I$ and
%
%
\begin{wrapfigure}{r}{0pt}
\xymatrix@!0@C-0.875pc@R-1.125pc{
Z
\ar[6,2]_-*{h}
\ar[0,4]^-*{f}
&&&&
Y
\ar[6,-2]^-*{g}
\\\\
&&
\mathrotate{36}{=}
\\\\\\\\
&&
X
}
\end{wrapfigure}
$\sum_{i \in I} r'_i = 1$.
For instance, $\RRp$ is normalizable.
Now, in $R\Map$ for normalizable $R$, we have the following fact (a proof is omitted since it is straightforward).

\begin{fact}\label{thm:marginal}
Suppose $R$ is normalizable.
Then a factorization of a surjection into surjections, $h = g \cmp f : Z \twoheadrightarrow Y \twoheadrightarrow X$, induces the following function $\phi$:
For any $R$-map $h^R$ on $h$, $\phi(h^R)$ is the (unique) $R$-map on $g$ through which $h^R$ factors (in $R\Map$);
that is, $h^R = f^R \cmp \phi(h^R)$ for some $R$-map $f^R$ on $f$.
More explicitly, $\phi(h^R)$ is defined by
%
%
$d^{\phi(h^R)}_s(t) = \sum_{u \in f^{-1}(t)} d^{h^R}_s(u)$
for $s \in X$ and $t \in Y$;
in other words, $\phi(h^R)$ is the \emph{marginal} of $h^R$ along the identification of states $u \in Z$ by the quotient map $f : Z \twoheadrightarrow Y$.
In addition, $\phi$ is a surjection from the $R$-maps on $h$ to those on $g$.
\end{fact}

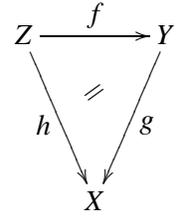
\begin{wrapfigure}{r}{0pt}
\xymatrix@!0@C-0.875pc@R-1.125pc{
\S_{AB}
\ar[3,1]_-*{\pi_{AB}}
\ar[0,4]^-*{p_\S}
&&&&
\S_A
\ar[6,-2]^-*{\pi_A}
\\\\
&&
\mathrotate{36}{=}
\\
&
L_{AB}
\ar[3,1]_-*{p_L}
\\\\\\
&&
L_A
}
\end{wrapfigure}
Let us apply this fact to the diagram in \eqref{eq:multipartite}, writing $p_\S : \S_{AB} \twoheadrightarrow \S_A$ and $p_L : L_{AB} \twoheadrightarrow L_A$ for the pair of projections.
Take $h = p_L \cmp \pi_{AB}$ and $g = \pi_A$, with $f = p_\S$.
Then $\phi(h^R)$ (on $\pi_A$) is the marginal of $h^R$ (on $p_L \cmp \pi_{AB}$) along the restriction of description from the bipartite system to Alice's.
Note that, however, this involves probabilities on $p_L$, that is, with which Bob chooses from measurements $b$ and $b'$.
Different probabilities on $p_L$ may lead to different $\phi(h^R)$---%
or maybe not, if the probabilities on $\pi_{AB}$ satisfy the no-signalling property.
More precisely, we have the following (in which we write $\pi_{AB}$ so as to connect to \eqref{eq:multipartite}, but the system can consist of any number of parties).

\begin{theorem}
An $R$-presheaf $\pi_{AB}^R$ on the presheaf $\pi_{AB}$ for a multipartite system satisfies no-signalling iff, for each pair of projections $p_\S$ and $p_L$, $\phi(\pi_{AB}^R \cmp p_L^R)$ is the same regardless of the choice of $R$-map $p_L^R$ on $p_L$.
\end{theorem}

\begin{proof}
First observe that, for each $t \in \S_A$, since ${p_\S}^{-1}(t) = \sum_{v \in {p_L}^{-1}(\pi_A(t))} ({p_\S}^{-1}(t) \cap {\pi_{AB}}^{-1}(v))$, we have
\begin{align}
\notag
d^{\phi(\pi_{AB}^R \cmp p_L^R)}_{\pi_A(t)}(t)
  = \sum_{u \in {p_\S}^{-1}(t)} d^{\pi_{AB}^R \cmp p_L^R}_{\pi_A(t)}(u)
& = \sum_{u \in {p_\S}^{-1}(t)} d^{p_L^R}_{\pi_A(t)}(\pi_{AB}(u)) \cdot d^{\pi_{AB}^R}_{\pi_{AB}(u)}(u) \\
\notag
& = \sum_{v \in {p_L}^{-1}({\pi_A(t)})} \sum_{u \in {p_\S}^{-1}(t) \cap {\pi_{AB}}^{-1}(v)} d^{p_L^R}_{\pi_A(t)}(v) \cdot d^{\pi_{AB}^R}_v(u) \\
\label{eq:no-signalling}
& = \sum_{v \in {p_L}^{-1}({\pi_A(t)})} d^{p_L^R}_{\pi_A(t)}(v) \cdot \sum_{u \in {p_\S}^{-1}(t) \cap {\pi_{AB}}^{-1}(v)} d^{\pi_{AB}^R}_v(u) .
\end{align}

Now suppose $\pi_{AB}^R$ satisfies no-signalling.
This means that each $t \in \S_A$ is assigned a real $e(t)$ such that every $v \in {p_L}^{-1}(\pi_A(t))$ satisfies
%
%
$\sum_{u \in {p_\S}^{-1}(t) \cap {\pi_{AB}}^{-1}(v)} d^{\pi_{AB}^R}_v(u) = e(t)$.
%
%
(Note that this is independent of choice of $R$-map on $p_L$.)
So, plugging this into \eqref{eq:no-signalling} we see that, for any $R$-map $p_L^R$ on $p_L$, each $t \in \S_A$ has
%
%
$d^{\phi(\pi_{AB}^R \cmp p_L^R)}_{\pi_A(t)}(t)
  = \sum_{v \in {p_L}^{-1}(\pi_A(t))} d^{p_L^R}_{\pi_A(t)}(v) \cdot e(t)
  = e(t)$,
%
%
which is independent of the choice of $p_L^R$.

On the other hand, suppose $\phi(\pi_{AB}^R \cmp p_L^R)$ is the same regardless of the choice of $p_L^R$.
Fix any $t \in \S_A$ and any $v_0, v_1 \in {p_L}^{-1}(\pi_A(t))$.
Obviously, for each $i = 0, 1$, there is an $R$-map $p_L^{Ri}$ on $p_L$ such that $d^{p_L^{Ri}}_{\pi_A(t)}(v_i) = 1$.
%
%
%
%
Then
%
%
$\sum_{u \in {p_\S}^{-1}(t) \cap {\pi_{AB}}^{-1}(v_0)} d^{\pi_{AB}^R}_{v_0}(u)
  = d^{\phi(\pi_{AB}^R \cmp p_L^{R0})}_{\pi_A(t)}(t)
  = d^{\phi(\pi_{AB}^R \cmp p_L^{R1})}_{\pi_A(t)}(t)
  = \sum_{u \in {p_\S}^{-1}(t) \cap {\pi_{AB}}^{-1}(v_1)} d^{\pi_{AB}^R}_{v_1}(u)$
%
%
by \eqref{eq:no-signalling}.
Since this holds for any pair of projections $p_\S$ and $p_L$, $\pi_{AB}^R$ satisfies no-signalling.
\end{proof}

\section{Stochastic Relational Presheaves}\label{sec:stochastic.relational.presheaf}

In \autoref{sec:presheaf} we saw how presheaves over trees---%
which are themselves trees---%
can be used as LTSs;
and in \autoref{sec:probability.presheaf} we saw how to add probabilities to such systems.
Generalizing this, this section obtains similarly labelled transition systems with probabilities that are however not trees.

\subsection{Relational Presheaves}\label{sec:relational.presheaf}

We first show how to implement LTSs of a non-tree shape using a presheaf-like structure.
The core idea in using presheaves as LTSs was the following, functorial one:
Let a tree $L$ represent a series of external choices;
assign to each stage in $L$ the set of possible states at that stage;
and connect states from different stages with internally chosen transitions.
This idea involves no intrinsic reason why this connection of transitions should be (reverse) functional, i.e., why the functor we take should be a presheaf.

In fact, here is a reason the functor we take should \emph{not} always be a presheaf.
Consider the following two objectives, each of which may, conceivably, be well motivated.
\begin{enumerate}
\item\label{item:why.relational.hilbert}
For our functor $\S$ from the tree $L$, we may like to take, as values $\S(x)$ for stages $x \in L$, the sets of states in Hilbert spaces instead of just any sets, to express quantum processes straightforwardly.
\item\label{item:why.relational.non.tree}
We may consider a non-deterministic process that involves both branching and colliding (so cannot be a tree, forward or backward).
In fact, when we do a quantum measurement $a$ in one basis and then another $a'$ in another basis, the system may transition from a state $s$ to $t_0$ (after $a$) to $u$ (after $a'$), but may also transition from $s$ to $t_1 \neq t_0$ (after $a$) to the same $u$ (after $a'$).
\end{enumerate}
The use of a presheaf, and in particular of functions $\S(e)$ for edges $e$ of $L$---%
which forces the transition system to be a tree---%
cannot accommodate both \eqref{item:why.relational.hilbert} and \eqref{item:why.relational.non.tree}.
To accommodate a non-tree as in \eqref{item:why.relational.non.tree} in a tree formalism, it is a standard technique to ``unfold'' or ``unravel'' the non-tree into a tree, duplicating the single state $u$ to $u_0$ following $t_0$ and $u_1$ following $t_1$.
This, however, does not go well with \eqref{item:why.relational.hilbert}, since the set $\S(x)$ encompassing $u_0$, $u_1$, and all the required duplicates may have to be much more complicated than just the set of states of a Hilbert space.
This is why we should at least sometimes let $\S(e)$, for edges $e$ of $L$, be relations in general rather than functions.
Then, in \eqref{item:why.relational.non.tree}, the state $s$ can be connected to both $t_0$ and $t_1$ while both $t_0$ and $t_1$ connected to $u$.

So, instead of the category $\Sets$ of sets and functions, we take the category $\Rel$ of sets and relations as the codomain of our functors (see \cite{car87} and \cite[esp.\ Ch.\ II]{fre90} for categorical characterizations of $\Rel$ and its generalizations).
For the sake of notation, let us enter

\begin{definition}
$\Rel$ is the category of sets and relations.
Its objects are sets, and its arrows from a set $X$ to another $Y$ are relations $f \subseteq X \times Y$, written $f : X \rel Y$ as well.
We write $s \xrightarrow{f} t$ instead of $(s, t) \in f$, and, identifying $f : X \rel Y$ with $f : X \to \pw(Y)$, sometimes write $f(s) = \{\, t \in Y \mid s \xrightarrow{f} t \,\}$.
The composition $g \cmp f : X \rel Z$ of $f : X \rel Y$ and $g : Y \rel Z$ is defined so that $s \xrightarrow{g \cmp f} u$ iff $s \xrightarrow{f} t \xrightarrow{g} u$ for some $t \in Y$.
\end{definition}

$\Rel$ is a dagger compact category.
Firstly, it has a $\dagger$ structure:
Any $f : X \rel Y$ has a unique opposite relation $f^\dagger : Y \rel X$, so that $s \xrightarrow{f^\dagger} t$ iff $t \xrightarrow{f} s$.
Also, even though the cartesian product is no longer the product in the categorical sense in $\Rel$, it is still a monoidal product $\otimes$.
In addition, the identification of $f : X \rel Y$ with $f : X \to \pw(Y)$ is just one aspect of the fact that $\Rel$ is the Kleisli category $\Kl(\pw)$ of the powerset monad $\pw$ on $\Sets$.
Now, let us finally provide

\begin{definition}
A \emph{relational presheaf} over a category $\mathbf{C}$ is a covariant functor from $\mathbf{C}$ to $\Rel$.%
\footnote{\strut%
Rosenthal \cite{ros96} defines a relational presheaf as a ``lax'' functor;
Soboci\'nski \cite{sob12} follows this ``lax'' definition in his account of relational presheaves as LTSs.
In contrast, I define a relational presheaf ``strongly''.
}
\end{definition}

So we generalize presheaves with relational presheaves as our LTSs.
We must note that relational presheaves are covariant and not contravariant.
Thus, given an edge $e = (x, y)$ of a tree $L$, the system's transition from states at stage $x$ to ones at stage $y$ is represented by a relation $\S(e) : \S(x) \rel \S(y)$ in a relational presheaf $\S : L \to \Rel$, whereas by a function $\S(e) : \S(y) \to \S(x)$ in a presheaf $\S : L^\op \to \Sets$.

It may be worth noting that, although relational presheaves over a tree of labels are themselves LTSs, they are also a generalization of the ordinary kind of LTSs in the following sense.
As Hermida \cite{her11} observes, given a set $L$ of labels, the (ordinary) transition systems labelled by $L$ are, in our terminology, the relational presheaves over the free monoid $L^\ast$ generated by $L$.
Our notion of relational presheaf as a LTS generalizes this by replacing $L^\ast$---%
a tree in which every (type of) edge is followed by every other (type of) edge---%
with a general tree, and permitting different stages to have different sets of states.

It is also worth noting that a small part of the presheaf-fibration equivalence (\autoref{thm:presheaf.fibration.equiv}) applies to relational presheaves, as relational presheaves over a tree $L$ can be regarded as ``open'' bundles over $L$:
The equivalence extends to an essentially surjective and faithful functor from the category of rooted and open bundles over a rooted tree $L$ to that of rooted relational presheaves over $L$.
(Again, we omit the definitions and proof in this abstract.)

\subsection{Adding Probabilities to Relational Presheaves}\label{sec:probability.relational.presheaf}

We added probabilities to presheaves as LTSs in \autoref{sec:probability.presheaf}.
In this subsection, we add probabilities to relational presheaves, which we introduced in \autoref{sec:relational.presheaf}.
This can be done by simply replacing the functional elements of the definitions in \autoref{sec:probability.presheaf} with relational elements.
(We should recall that, in the generalization given in \autoref{sec:relational.presheaf}, a relation $f : X \rel Y$ generalizes a function $f : Y \to X$ of the opposite direction.)

\begin{definition}
We define an \emph{$R$-relation} from a set $X$ to another $Y$ as an ``entire'' relation $f : X \rel Y$ (i.e., such that each $s \in X$ has some $t \in Y$ with $s \xrightarrow{f} t$;
note that, unlike the case of $R$-maps, we do not flip the direction of $f$ for $R$-relations) equipped with, for each $s \in X$, an $R$-distribution $d^f_s$ on $Y$ with support
\begin{gather}\label{eq:R.rel.support}
\supp{d^f_s} \subseteq f(s) .
\end{gather}
(We say that an $R$-relation is \emph{on} its underlying relation.)
%
%
Given two $R$-relations $f : X \rel Y$ and $g : Y \rel Z$, let their composition $g \cmp f : X \rel Z$ have, for each $s \in X$, an $R$-distribution $d^{g \cmp f}_s$ on $Z$ such that
\begin{gather}\label{eq:R.rel.composition}
d^{g \cmp f}_s(u) = \sum_{t \in Y} d^f_s(t) \cdot d^g_t(u) .
\end{gather}
Write $R\RRel$ for the category of sets and $R$-relations.
(It should be clear that the unique $R$-relation on the identity relation $1_X : X \rel X$ is the identity on $X$ in $R\RRel$.)
\end{definition}

This notion of $R$-relation is closely related to that of \emph{stochastic map}.
We discuss this relationship in \autoref{sec:stochastic.relation.math};
it will be significant to the discussion that \eqref{eq:R.rel.support} has ``$\subseteq$'' as opposed to ``$=$''.

Let us compare the equation \eqref{eq:R.rel.composition} with the one \eqref{eq:R.map.composition} for $R$-maps.
For $R$-maps $f : Z \to Y$ and $g : Y \to X$, there is at most one state $t \in Y$ through which the system may transition from a given $s \in X$ to a given $u \in Z$;
so the probability of the transition from $s$ to $u$ is just the probability of this particular path, given by the product of the two transitions, from $s$ to $t$ and from $t$ to $u$.
In contrast, for $R$-relations $f : X \rel Y$ and $g : Y \rel Z$, there can be many paths through which the system may transition from $s \in X$ to $u \in Z$;
yet, since these paths are mutually exclusive, we can just sum their probabilities up to obtain the probability of the transition from $s$ to $u$.
Lastly, enter

\begin{definition}
An \emph{$R$-relational presheaf} over a category $\mathbf{C}$ is a covariant functor from $\mathbf{C}$ to $R\RRel$.
(We say that an $R$-relational presheaf is \emph{on} its underlying relational presheaf.)
\end{definition}

This definition provides a structure that integrates the three frameworks \eqref{item:presheaf}--\eqref{item:stoch} mentioned in Introduction:
An $R$-relational presheaf $\S : L \to R\RRel$ over a tree $L$ forms an $L$-LTS in which internal choices take place with probabilities and possibly in a non-tree fashion.

\begin{example}
Let a tree $L$ represent a branching family of series of quantum measurements, gates, and other operations that can be performed.
Then, for stages $x \in L$, let $\S(x)$ be sets of states in (possibly, though not necessarily, identical) Hilbert spaces, and, for each edge $e = (x, y)$ of $L$, let $\S(e) : \S(x) \rel \S(y)$ be the $\RRp$-relation that models the operation $e$ in Hilbert-space terms, such as projections (branching with probabilities) to the suitable measurement basis.
If $L$ is moreover a free monoid and $\S(x)$ are all identical (as in Hermida's \cite{her11} formulation of transition systems mentioned in \autoref{sec:relational.presheaf}), models amount essentially to ones given in Baltag and Smets \cite{bal06}.
\end{example}

This example gives a straightforward representation of quantum protocols.
So it is not surprising at all that we can find non-local or contextual behaviors in such representations.
Yet, using more general values than Hilbert spaces, $R$-relational presheaves can model not only the presence but also the absence of non-locality and contextuality, and indeed characterize contextuality, as we will see in \autoref{sec:contextuality}.

\subsection{Relation to Other Work and Formulations}\label{sec:stochastic.relation.math}

The notion of $R$-relation is closely related to that of \emph{stochastic map}, or equivalently to Kleisli maps of the \emph{distribution monad}.%
\footnote{\strut%
I thank an anonymous referee for his/her comments regarding the relation between $R\RRel$ and $\Kl(\D_\RRp)$, which prompted me to write this subsection as a reply.
}
A stochastic map from a set $X$ to another $Y$ is an $X$-indexed family of $\RRp$-distributions on $Y$, with the composition defined exactly by \eqref{eq:R.rel.composition}.
This can also be rewritten using

\begin{definition}\label{def:R.dist.monad}
Given any set $X$, write $\D_R(X)$ for the set of $R$-distributions on $X$.
This gives rise to the $R$-distribution functor $\D_R : \Sets \to \Sets$ (see \cite{jac10} as well as \cite[\Sect 2.3]{abr11}), which is in fact a monad on $\Sets$ (see \cite{jac10}).
\end{definition}

Then the stochastic maps $f$ from a set $X$ to another $Y$ are exactly the functions $f : X \to \D_\RRp(Y)$, the Kleisli maps of $\D_\RRp$.
Moreover the Kleisli composition amounts to \eqref{eq:R.rel.composition}, and so the category $\Stoch$ of sets and stochastic maps is the Kleisli category $\Kl(\D_\RRp)$ of $\D_\RRp$ (see \cite[\Sect 2]{jac11}).

This is closely related to $R\RRel$, but not exactly the same (aside from $R$ generalizing $\RRp$):
In short, an $\RRp$-relation on a relation $f$ is a stochastic map with an extra piece of information, namely, the underlying relation $f$.
To express this formally, consider the following subfunctor of $\pw \times \D_R : \Sets \to \Sets$.
\begin{gather*}
\T : X \mapsto \sum_{\S \in \pw(X)} \D_R(\S) = \{\, (\S, d) \in \pw(X) \times \D_R(X) \mid d \in \D_R(\S) \,\} .
\end{gather*}
(We identify $d \in \D_R(X)$ and $d \in \D_R(\S)$, as long as $\supp{d} \subseteq \S, X$.)
Then the $R$-relations $f$ from a set $X$ to another $Y$ are exactly the functions $f : X \to \T(Y)$, with a $\pw(Y)$ component.
The two sets $\T(Y)$ and $\D_R(Y)$ are related by the projection $p : (\S, d) \mapsto d$ and a section $s : d \mapsto (\supp{d}, d)$, but $s \cmp p \neq 1$ since we have ``$\subseteq$'' as opposed to ``$=$'' in \eqref{eq:R.rel.support}.
Thus an $R$-relation $f : X \to \T(Y)$ carries properly more information, of the underlying relation, than a stochastic map $f : X \to \D_R(Y)$.
More categorically put, postcomposing $p$ and $s$ with Kleisli maps gives a retraction and a section of categories so that

\begin{fact}
$\Stoch = \Kl(\D_\RRp)$ is a retract of $\RRp\RRel$, but the retraction is not faithful.
\end{fact}

The extra piece of information may appear redundant, as long as we are concerned with probabilities of transitions;
yet that piece of information sometimes proves useful.
In such a model as in \eqref{eq:change.of.base} or \eqref{eq:multipartite}, the underlying relational presheaf $\S$ describes the ``logical'' constraint of which states can be ``logically'' connected to which states;
for instance, on the left of \eqref{eq:change.of.base}, state $00$ can follow $0$ but cannot $1$.
When we add the ``physical'' information of probabilities to $\S$ by taking an $R$-relational presheaf on $\S$, the ``logical'' information is sometimes entailed by supports, but not always so:
If the edge connecting states $0$ and $00$ in \eqref{eq:change.of.base} has probability $0$, then the support cannot tell us whether state $00$ can ``logically'' follow state $0$ or $1$.
It is useful to retain the ``logical'' constraint so as to consider a family of physical models satisfying it, as opposed to just one model---%
it is as useful as having a table of $4 \times 4$ entries that accommodates a family of probability assignments to outcomes in a $(2, 2, 2)$-scenario.
And for this purpose we need to retain the underlying relations, hence using $R\RRel$ as opposed to $\Stoch$.

Lastly, it may be useful to note that the unit and multiplication of the cartesian product monad $\pw \times \D_R$ restrict to the subfunctor $T$,%
\footnote{\strut%
See Definition 2.1.2 of \cite{man07} for a concrete description of a cartesian product monad and its unit and multiplication.
}
and that the composition in $R\RRel$ is the Kleisli composition of $T$;
thus

\begin{fact}
$\T$ is a monad on $\Sets$, and $R\RRel$ is its Kleisli category $\Kl(\T)$.
\end{fact}

This puts $R\RRel$ in the tradition \cite[etc.]{law62,gir82,pan98,pan99,dob07,jac11}\ of using algebras for monads to represent stochastic relations.

\section{Dynamic Logic for Contextuality}\label{sec:contextuality}

So far we have laid out $R$-presheaves and $R$-relational presheaves as labelled and stochastic transition systems.
Now we demonstrate that these models are good enough for representing essential features of stochastic dynamics such as shown in quantum systems, by showing that they can characterize non-locality and contextuality;
in fact, the dynamic logic of those transition systems is expressive enough to express this characterization in logical terms.

\subsection{Deterministic Hidden-Variable Models}\label{sec:hidden.variable}

In their sheaf-theoretic approach to non-locality and contextuality, Abramsky and Brandenburger \cite{abr11} provided a characterization of non-locality and contextuality in terms of ``global sections'' of certain presheaves;
see \cite[esp.\ \Sect 3 and \Sect 8]{abr11}.
We can ``translate'' this characterization into our setting of stochastic relational presheaves as follows.

Suppose we have an $\RRp$-presheaf representing an ``empirical model'' for a $(n, k, \ell)$-scenario that satisfies no-signalling in the sense of \autoref{sec:no.signalling}.
(The characterization given in \cite{abr11} is more general than just about $(n, k, \ell)$-scenarios, though I only take $(n, k, \ell)$-scenarios here.
We can translate the characterization in full generality, but omit it in this abstract.)
As an example, let us take an $\RRp$-presheaf $E$ on the presheaf $\S_{AB}$ in \eqref{eq:multipartite} (and assume no-signalling).
Then $E$ is realized by a (factorizable) hidden-variable model if and only if it has a ``global section'' (Theorem 8.1 of \cite{abr11})---%
meaning, in our terms, that there exists an $\RRp$-relational presheaf $H$ on the relational presheaf $\S^h_{AB}$ in
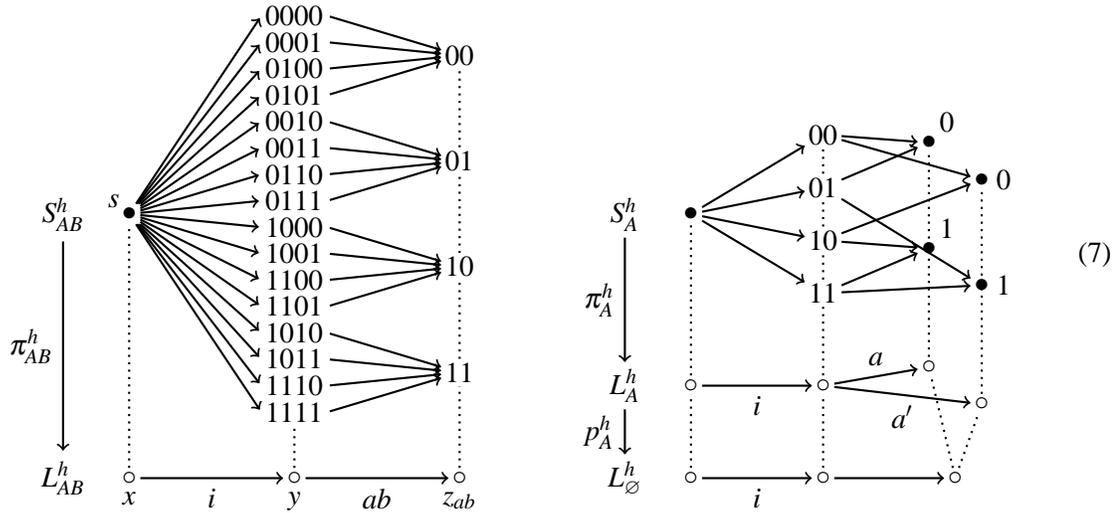
\begin{align}
\label{eq:hidden.variable}
\begin{tikzpicture}[x=50pt,y=40pt,thick,label distance=-0.25em,baseline=(current bounding box.center)]
\node (m) at (1.25,0) {};
\node (T) at (0,2.5) {};
\node (u) at (0,0.25) {};
\node (d) at (0,-0.25) {};
\node [inner sep=0.1em,label=below:{$x$}] (L0O) at (0,0) {$\circ$};
\node [inner sep=0.1em,label=below:{$y$}] (L0Oa) at ($ (L0O) + (m) $) {$\circ$};
\node [inner sep=0.1em,label=below:{$z_{ab}$}] (L0Oab) at ($ (L0Oa) + (m) $) {$\circ$};
\node [inner sep=0.1em,label={[label distance=-0.5em]105:{$s$}}] (S0O) at ($ (L0O) + (T) $) {$\bullet$};
\node [inner sep=0.1em] (S0Oa7) at ($ (S0O) + (m) + (u)!0.5!(0,0) $) {$0111$};
\node [inner sep=0.1em] (S0Oa6) at ($ (S0Oa7) + (u) $) {$0110$};
\node [inner sep=0.1em] (S0Oa5) at ($ (S0Oa6) + (u) $) {$0011$};
\node [inner sep=0.1em] (S0Oa4) at ($ (S0Oa5) + (u) $) {$0010$};
\node [inner sep=0.1em] (S0Oa3) at ($ (S0Oa4) + (u) $) {$0101$};
\node [inner sep=0.1em] (S0Oa2) at ($ (S0Oa3) + (u) $) {$0100$};
\node [inner sep=0.1em] (S0Oa1) at ($ (S0Oa2) + (u) $) {$0001$};
\node [inner sep=0.1em] (S0Oa0) at ($ (S0Oa1) + (u) $) {$0000$};
\node [inner sep=0.1em] (S0Oa8) at ($ (S0Oa7) + (d) $) {$1000$};
\node [inner sep=0.1em] (S0Oa9) at ($ (S0Oa8) + (d) $) {$1001$};
\node [inner sep=0.1em] (S0Oa10) at ($ (S0Oa9) + (d) $) {$1100$};
\node [inner sep=0.1em] (S0Oa11) at ($ (S0Oa10) + (d) $) {$1101$};
\node [inner sep=0.1em] (S0Oa12) at ($ (S0Oa11) + (d) $) {$1010$};
\node [inner sep=0.1em] (S0Oa13) at ($ (S0Oa12) + (d) $) {$1011$};
\node [inner sep=0.1em] (S0Oa14) at ($ (S0Oa13) + (d) $) {$1110$};
\node [inner sep=0.1em] (S0Oa15) at ($ (S0Oa14) + (d) $) {$1111$};
\node [inner sep=0.1em] (S0Oab0) at ($ (S0Oa0)!0.5!(S0Oa3) + (m) $) {$00$};
\node [inner sep=0.1em] (S0Oab1) at ($ (S0Oa4)!0.5!(S0Oa7) + (m) $) {$01$};
\node [inner sep=0.1em] (S0Oab2) at ($ (S0Oa8)!0.5!(S0Oa11) + (m) $) {$10$};
\node [inner sep=0.1em] (S0Oab3) at ($ (S0Oa12)!0.5!(S0Oa15) + (m) $) {$11$};
\node [inner sep=0em] (S0Oa0L) at ($ (S0Oa0) + (-0.25,0) $) {};
\node [inner sep=0em] (S0Oa1L) at ($ (S0Oa1) + (-0.25,0) $) {};
\node [inner sep=0em] (S0Oa2L) at ($ (S0Oa2) + (-0.25,0) $) {};
\node [inner sep=0em] (S0Oa3L) at ($ (S0Oa3) + (-0.25,0) $) {};
\node [inner sep=0em] (S0Oa4L) at ($ (S0Oa4) + (-0.25,0) $) {};
\node [inner sep=0em] (S0Oa5L) at ($ (S0Oa5) + (-0.25,0) $) {};
\node [inner sep=0em] (S0Oa6L) at ($ (S0Oa6) + (-0.25,0) $) {};
\node [inner sep=0em] (S0Oa7L) at ($ (S0Oa7) + (-0.25,0) $) {};
\node [inner sep=0em] (S0Oa8L) at ($ (S0Oa8) + (-0.25,0) $) {};
\node [inner sep=0em] (S0Oa9L) at ($ (S0Oa9) + (-0.25,0) $) {};
\node [inner sep=0em] (S0Oa10L) at ($ (S0Oa10) + (-0.25,0) $) {};
\node [inner sep=0em] (S0Oa11L) at ($ (S0Oa11) + (-0.25,0) $) {};
\node [inner sep=0em] (S0Oa12L) at ($ (S0Oa12) + (-0.25,0) $) {};
\node [inner sep=0em] (S0Oa13L) at ($ (S0Oa13) + (-0.25,0) $) {};
\node [inner sep=0em] (S0Oa14L) at ($ (S0Oa14) + (-0.25,0) $) {};
\node [inner sep=0em] (S0Oa15L) at ($ (S0Oa15) + (-0.25,0) $) {};
\node [inner sep=0em] (S0Oa0R) at ($ (S0Oa0) + (0.25,0) $) {};
\node [inner sep=0em] (S0Oa1R) at ($ (S0Oa1) + (0.25,0) $) {};
\node [inner sep=0em] (S0Oa2R) at ($ (S0Oa2) + (0.25,0) $) {};
\node [inner sep=0em] (S0Oa3R) at ($ (S0Oa3) + (0.25,0) $) {};
\node [inner sep=0em] (S0Oa4R) at ($ (S0Oa4) + (0.25,0) $) {};
\node [inner sep=0em] (S0Oa5R) at ($ (S0Oa5) + (0.25,0) $) {};
\node [inner sep=0em] (S0Oa6R) at ($ (S0Oa6) + (0.25,0) $) {};
\node [inner sep=0em] (S0Oa7R) at ($ (S0Oa7) + (0.25,0) $) {};
\node [inner sep=0em] (S0Oa8R) at ($ (S0Oa8) + (0.25,0) $) {};
\node [inner sep=0em] (S0Oa9R) at ($ (S0Oa9) + (0.25,0) $) {};
\node [inner sep=0em] (S0Oa10R) at ($ (S0Oa10) + (0.25,0) $) {};
\node [inner sep=0em] (S0Oa11R) at ($ (S0Oa11) + (0.25,0) $) {};
\node [inner sep=0em] (S0Oa12R) at ($ (S0Oa12) + (0.25,0) $) {};
\node [inner sep=0em] (S0Oa13R) at ($ (S0Oa13) + (0.25,0) $) {};
\node [inner sep=0em] (S0Oa14R) at ($ (S0Oa14) + (0.25,0) $) {};
\node [inner sep=0em] (S0Oa15R) at ($ (S0Oa15) + (0.25,0) $) {};
\draw [->] (L0O) -- (L0Oa) node [pos=0.5,below] {$i$};
\draw [->] (L0Oa) -- (L0Oab) node [pos=0.5,below] {$ab$};
\draw [->] (S0O) -- (S0Oa0L);
\draw [->] (S0O) -- (S0Oa1L);
\draw [->] (S0O) -- (S0Oa8L);
\draw [->] (S0O) -- (S0Oa9L);
\draw [->] (S0O) -- (S0Oa2L);
\draw [->] (S0O) -- (S0Oa4L);
\draw [->] (S0O) -- (S0Oa6L);
\draw [->] (S0O) -- (S0Oa3L);
\draw [->] (S0O) -- (S0Oa7L);
\draw [->] (S0O) -- (S0Oa11L);
\draw [->] (S0O) -- (S0Oa15L);
\draw [->] (S0O) -- (S0Oa5L);
\draw [->] (S0O) -- (S0Oa10L);
\draw [->] (S0O) -- (S0Oa12L);
\draw [->] (S0O) -- (S0Oa13L);
\draw [->] (S0O) -- (S0Oa14L);
\draw [->] (S0Oa4R) -- (S0Oab1);
\draw [->] (S0Oa5R) -- (S0Oab1);
\draw [->] (S0Oa6R) -- (S0Oab1);
\draw [->] (S0Oa7R) -- (S0Oab1);
\draw [->] (S0Oa8R) -- (S0Oab2);
\draw [->] (S0Oa9R) -- (S0Oab2);
\draw [->] (S0Oa10R) -- (S0Oab2);
\draw [->] (S0Oa11R) -- (S0Oab2);
\draw [->] (S0Oa12R) -- (S0Oab3);
\draw [->] (S0Oa13R) -- (S0Oab3);
\draw [->] (S0Oa14R) -- (S0Oab3);
\draw [->] (S0Oa15R) -- (S0Oab3);
\draw [->] (S0Oa0R) -- (S0Oab0);
\draw [->] (S0Oa1R) -- (S0Oab0);
\draw [->] (S0Oa2R) -- (S0Oab0);
\draw [->] (S0Oa3R) -- (S0Oab0);
\draw [dotted] (L0O) -- (S0O);
\draw [dotted] (L0Oa) -- (S0Oa15);
\draw [dotted] (L0Oab) -- (S0Oab3) -- (S0Oab2) -- (S0Oab1) -- (S0Oab0);
\node (L0) at ($ (L0O) + (-0.5,0) $) {$L^h_{AB}$};
\node (S0) at ($ (S0O) + (-0.5,0) $) {$\S^h_{AB}$};
\draw [->] (S0) -- (L0) node [pos=0.5,left] {$\pi^h_{AB}$};
\node (i) at (1,0) {};
\node (l) at (-0.2,0.175) {};
\node (r) at (0.2,-0.175) {};
%
\node [inner sep=0.1em] (L0O) at (4.25,0) {$\circ$};
\node [inner sep=0.1em] (L0Oi) at ($ (L0O) + (i) $) {$\circ$};
\node [inner sep=0.1em] (L0Oii) at ($ (L0Oi) + (i) $) {$\circ$};
\draw [->] (L0O) -- (L0Oi) node [pos=0.5,below] {$i$};
\draw [->] (L0Oi) -- (L0Oii);
\node [inner sep=0.1em] (LAO) at ($ (L0O) + (0,0.875) $) {$\circ$};
\node [inner sep=0.1em] (LAOi) at ($ (LAO) + (i) $) {$\circ$};
\node [inner sep=0.1em] (LAOia) at ($ (LAOi) + (i) + (l) $) {$\circ$};
\node [inner sep=0.1em] (LAOia') at ($ (LAOi) + (i) + (r) $) {$\circ$};
\draw [->] (LAO) -- (LAOi) node [pos=0.5,below] {$i$};
\draw [->] (LAOi) -- (LAOia) node [pos=0.5,above] {$a$};
\draw [->] (LAOi) -- (LAOia') node [pos=0.5,below] {$a'$};
\node [inner sep=0.1em] (SAO) at ($ (LAO) + (0,1.625) $) {$\bullet$};
\node [inner sep=0.1em] (SAOi00) at ($ (SAO) + (i) + (u) + (u) + (u) $) {$00$};
\node [inner sep=0.1em] (SAOi01) at ($ (SAO) + (i) + (u) $) {$01$};
\node [inner sep=0.1em] (SAOi10) at ($ (SAO) + (i) + (d) $) {$10$};
\node [inner sep=0.1em] (SAOi11) at ($ (SAO) + (i) + (d) + (d) + (d) $) {$11$};
\node [inner sep=0.1em,label={[label distance=-0.5em]45:{$0$}}] (SAOia0) at ($ (SAO) + (i) + (i) + (l) + (u) + (u) $) {$\bullet$};
\node [inner sep=0.1em,label={[label distance=-0.5em]45:{$1$}}] (SAOia1) at ($ (SAO) + (i) + (i) + (l) + (d) + (d) $) {$\bullet$};
\node [inner sep=0.1em,label=right:{$0$}] (SAOia'0) at ($ (SAO) + (i) + (i) + (r) + (u) + (u) $) {$\bullet$};
\node [inner sep=0.1em,label=right:{$1$}] (SAOia'1) at ($ (SAO) + (i) + (i) + (r) + (d) + (d) $) {$\bullet$};
\draw [->] (SAO) -- (SAOi00);
\draw [->] (SAO) -- (SAOi01);
\draw [->] (SAO) -- (SAOi10);
\draw [->] (SAO) -- (SAOi11);
\draw [->] (SAOi00) -- (SAOia0);
\draw [->] (SAOi01) -- (SAOia0);
\draw [->] (SAOi10) -- (SAOia1);
\draw [->] (SAOi11) -- (SAOia1);
\draw [->] (SAOi00) -- (SAOia'0);
\draw [->] (SAOi01) -- (SAOia'1);
\draw [->] (SAOi10) -- (SAOia'0);
\draw [->] (SAOi11) -- (SAOia'1);
\draw [dotted] (L0O) -- (LAO) -- (SAO);
\draw [dotted] (L0Oi) -- (LAOi) -- (SAOi11) -- (SAOi10) -- (SAOi01) -- (SAOi00);
\draw [dotted] (L0Oii) -- (LAOia) -- (SAOia0);
\draw [dotted] (L0Oii) -- (LAOia') -- (SAOia'0);
\node [inner sep=0.2em] (L0) at ($ (L0O) + (-0.5,0) $) {$L^h_\varnothing$};
\node [inner sep=0.2em] (LA) at ($ (LAO) + (-0.5,0) $) {$L^h_A$};
\node [inner sep=0.2em] (SA) at ($ (SAO) + (-0.5,0) $) {$\S^h_A$};
\draw [->] (SA) -- (LA) node [pos=0.5,left] {$\pi^h_A$};
\draw [->] (LA) -- (L0) node [pos=0.5,left] {$p^h_A$};
\end{tikzpicture}
\end{align}
(complete the picture by adding edges $ab'$, $a'b$, and $a'b'$ to $L^h_{AB}$) from which $E$ is obtained by forgetting the middle stage $y$ with the change-of-base operation as on the left of \eqref{eq:change.of.base}, that is, $E = H \cmp m^\op$ for the embedding $m : L_{AB} \rightarrowtail L^h_{AB}$ that omits $y$.
Here $E = H \cmp m^\op$ means that $E(ab) = H(ab \cmp i) = H(ab) \cmp H(i)$, and hence that, by \eqref{eq:R.rel.composition},
\begin{gather}
\label{eq:averaging.over}
d^{E(ab)}_s(u) = \sum_{t \in H(y)} d^{H(i)}_s(t) \cdot d^{H(ab)}_t(u) .
\end{gather}
This is exactly to ``reproduce the empirically observed probabilities [$d^{E(ab)}_s$] by \emph{averaging} over the hidden variables with respect to the distribution [$d^{H(i)}_s$]'' (\cite{abr11}, p.\ 11).

From this characterization, the following features of $H$ should be obvious:
The set $H(y)$, which is forgotten in $E$, is a set of latent ``instruction sets'' (see \cite{mer81});
moreover, they are \emph{deterministic}, as $H(ab)$ is an $\RRp$-relation on a \emph{function}, as opposed to just any relation, from $H(y)$ to $H(z_{ab}) = E(z_{ab})$.
Thus, the contextuality in a labelled and stochastic transition system $E$ amounts to the failure of $E$ to have such a deterministic hidden-variable model $H$.
A little more formally,

\begin{theorem}\label{thm:global.section}
For an empirical model $E$ (in the sense of \cite{abr11}), the following are equivalent.
\begin{enumerate}
\item\label{item:global.section.realized}
$E$ has a realization by a factorizable hidden-variable model.
\item\label{item:global.section.exists}
$E$ has a global section.
\item\label{item:global.section.forgot}
The $\RRp$-relational presheaf for $E$ is obtained by forgetting the middle stage of a deterministic hidden-variable model.
\end{enumerate}
\end{theorem}

\begin{proof}
``\eqref{item:global.section.realized} iff \eqref{item:global.section.exists}'' is Theorem 8.1 of \cite{abr11}.
``\eqref{item:global.section.exists} iff \eqref{item:global.section.forgot}'' is essentially due to the fact that the equation for ``averaging over'' in \cite{abr11} (p.\ 11) is identical to \eqref{eq:averaging.over}.
\end{proof}

Note that the underlying relational presheaf $\S^h_{AB}$ of $H$ (or any general ones for $(n, k, \ell)$-scenarios) is not provided \textit{ad hoc}, but canonically obtained, in the manner of \eqref{eq:multipartite}, as the fibered cartesian product $\S^h_A \otimes_{L^h_\varnothing} \S^h_B$ of the obvious hidden-variable models $\S^h_A : L^h_A \to \Rel$ for Alice (as in \eqref{eq:hidden.variable} above) and for Bob.

\begin{wrapfigure}{r}{0pt}
\begin{tikzpicture}[x=40pt,y=40pt,thick,label distance=-0.25em,baseline=(current bounding box.center)]
\node (m) at (1,0) {};
\node (u) at (0,0.15) {};
\node (d) at (0,-0.15) {};
\node [inner sep=0.1em] (SDO) at (0,0) {$\bullet$};
\node [inner sep=0.1em] (SDOi0) at ($ (SDO) + (m) + (u) $) {$\bullet$};
\node [inner sep=0.1em] (SDOi1) at ($ (SDO) + (m) + (d) $) {$\bullet$};
\node [inner sep=0.1em] (SDOi0m) at ($ (SDOi0) + (m) $) {$\bullet$};
\node [inner sep=0.1em] (SDOi1m) at ($ (SDOi1) + (m) $) {$\bullet$};
\draw [->] (SDO) -- (SDOi0);
\draw [->] (SDO) -- (SDOi1);
\draw [->] (SDOi0) -- (SDOi0m);
\draw [->] (SDOi1) -- (SDOi1m);
\node [inner sep=0.1em] (SNO) at (0,-0.75) {$\bullet$};
\node [inner sep=0.1em] (SNOi) at ($ (SNO) + (m) $) {$\bullet$};
\node [inner sep=0.1em] (SNOim0) at ($ (SNOi) + (m) + (u) $) {$\bullet$};
\node [inner sep=0.1em] (SNOim1) at ($ (SNOi) + (m) + (d) $) {$\bullet$};
\draw [->] (SNO) -- (SNOi);
\draw [->] (SNOi) -- (SNOim0);
\draw [->] (SNOi) -- (SNOim1);
\end{tikzpicture}
\end{wrapfigure}
To extract an essential idea from the discussion so far, contextuality means, in transitional terms, that a model is inconsistent with the first shape of branching to the right (in which the system internally chooses from latent instruction sets before external choices are made), but has to have the second shape (in which the system internally chooses outcomes when external choices are made).
And the distinction between these two shapes is one of the things modal logic is good at.
Thus we carry on to consider the modal logic of our labelled and stochastic transition systems.

\subsection{Dynamic Logic of Stochastic Relational Presheaf Models}\label{sec:dynamic.logic}

We lay out here how to use $R$-relational presheaves as a semantic structure for modal, dynamic logic.
It turns out that the logic it gives rise to is expressive enough to capture in logical terms the characterization of contextuality we saw in \autoref{sec:hidden.variable}.
(See \cite{har00} for general exposition of dynamic logic.
A modal logic of stochastic relations expressed by algebras for a monad is also found in \cite{dob07}.)

Let us fix some (propositional) language;
for our purposes it needs to have $\wedge$ and $\lnot$.
Then we fix a set of labels $e$ of transition (for instance, we use labels $a$, $b$, $ab$, $ab'$, etc., for a measurement scenario of $(2, 2, 2)$-type).
For each such label $e$, we add ``dynamic modalities'' $\nec{e}$ and $\pos{e}$ to the language;
we may also like to use probability modalities $P({-} \mid e) \gtreqqless r$ for reals $r$.
Since we take the base logic to be classical, $\pos{e}$ can be defined as $\lnot \nec{e} \lnot$, and $\top$, $\vee$, biconditional $\leftrightarrow$ and exclusive disjunction $\oplus$ can be defined as usual.
So we put
\begin{gather*}
\varphi ::= p \mid \varphi \wedge \varphi \mid \lnot \varphi \mid \nec{e} \varphi \mid P(\varphi \mid e) > r \mid P(\varphi \mid e) = r \mid P(\varphi \mid e) < r
\end{gather*}
for propositional letters $p$ and the labels $e$.
For the sake of application to the contextuality in quantum measurements, we let each label be a measurement context (i.e., a jointly performable set of measurements), and each propositional letter have the form $a = k$ for a measurement $a$ and an outcome $k$ of $a$.
(We will mention this application shortly in \autoref{sec:contextuality};
the semantics laid out in the remainder of this subsection can apply equally to other languages of the sort just defined.)

\begin{wrapfigure}{r}{0pt}
\begin{tikzpicture}[x=40pt,y=40pt,thick,label distance=-0.25em,baseline=(current bounding box.center)]
\node (i) at (1,0) {};
\node (j) at (1.5,0) {};
\node (l) at (-0.25,0.15) {};
\node (r) at (0.25,-0.15) {};
\node (uu) at (0,0.2) {};
\node (dd) at (0,-0.2) {};
\node [inner sep=0.1em] (LAO) at (0,0) {$\circ$};
\node [inner sep=0.1em] (LAOa) at ($ (LAO) + (j) + (l) $) {$\circ$};
\node [inner sep=0.1em] (LAOa') at ($ (LAO) + (j) + (r) $) {$\circ$};
\draw [->] (LAO) -- (LAOa) node [pos=0.5,above] {$ab$};
\draw [->] (LAO) -- (LAOa') node [pos=0.5,below] {$ab'$};
\node [inner sep=0.1em,label=above:{$s$}] (SAO) at ($ (LAO) + (0,1.5) $) {$\bullet$};
\node [inner sep=0.1em,label=right:{$00$}] (SAOa00) at ($ (SAO) + (j) + (l) + (uu) + (uu) + (uu) $) {$\bullet$};
\node [inner sep=0.1em] (SAOa01) at ($ (SAO) + (j) + (l) + (uu) $) {$\bullet$};
\node [inner sep=0.1em] (SAOa10) at ($ (SAO) + (j) + (l) + (dd) $) {$\bullet$};
\node [inner sep=0.1em] (SAOa11) at ($ (SAO) + (j) + (l) + (dd) + (dd) + (dd) $) {$\bullet$};
\node [inner sep=0.1em,label=right:{$00$}] (SAOa'00) at ($ (SAO) + (j) + (r) + (uu) + (uu) + (uu) $) {$\bullet$};
\node [inner sep=0.1em,label=right:{$01$}] (SAOa'01) at ($ (SAO) + (j) + (r) + (uu) $) {$\bullet$};
\node [inner sep=0.1em,label=right:{$10$}] (SAOa'10) at ($ (SAO) + (j) + (r) + (dd) $) {$\bullet$};
\node [inner sep=0.1em,label=right:{$11$}] (SAOa'11) at ($ (SAO) + (j) + (r) + (dd) + (dd) + (dd) $) {$\bullet$};
\draw [->] (SAO) -- (SAOa00);
\draw [->] (SAO) -- (SAOa01);
\draw [->] (SAO) -- (SAOa10);
\draw [->] (SAO) -- (SAOa11);
\draw [->] (SAO) -- (SAOa'00);
\draw [->] (SAO) -- (SAOa'01);
\draw [->] (SAO) -- (SAOa'10);
\draw [->] (SAO) -- (SAOa'11);
\node [inner sep=0.2em] (LA) at ($ (LAO) + (-0.5,0) $) {$L_0$};
\node [inner sep=0.2em] (SA) at ($ (SAO) + (-0.5,0) $) {$E$};
\draw [->] (SA) -- (LA) node [pos=0.5,left] {$\pi$};
\node [inner sep=0.1em] (L0O) at ($ (LAO) + (0,-0.875) $) {$\circ$};
\node [inner sep=0.1em] (L0Oi) at ($ (L0O) + (j) $) {$\circ$};
\draw [->] (L0O) -- (L0Oi) node [pos=0.5,below] {$a$};
\draw [dotted] (L0O) -- (LAO) -- (SAO);
\draw [dotted] (L0Oi) -- (LAOa) -- (SAOa00);
\draw [dotted] (L0Oi) -- (LAOa') -- (SAOa'00);
\node [inner sep=0.2em] (L0A) at ($ (LA) + (0,-0.875) $) {$L_1$};
\draw [->] (LA) -- (L0A) node [pos=0.5,left] {$p$};
\end{tikzpicture}
\end{wrapfigure}
For this modal language, $\RRp$-relational presheaves provide models.
Firstly, the labels need interpreting in trees of labels.
We take a family of trees---%
but not necessarily a single tree---%
and fibrations among them with an initial tree $L_0$, so that each label $e$ is an edge of one of the trees.
For instance, the $(2, 2, 2)$-scenario of Alice and Bob described in \eqref{eq:multipartite} has four trees $L_{-}$ of labels and fibrations among them, with $L_{AB}$ initial;
the picture to the right describes $L_0 = L_{AB}$, $L_1 = L_A$, and the fibration $p : L_{AB} \to L_A$.
Labels $ab$ and $ab'$ lie in $L_0$, which takes both Alice and Bob as external to the system;
$a$ lies in $L_1$, which takes Alice as external but Bob as internal.
In general, $p$ may fail to be a fibration or a function to $L_1$ (because, e.g., $L_0$ may have edges $ab$, $bc$, $ca$ while $L_1$ has only $a$, so that $bc$ cannot be projected down to any edge in $L_1$);
but it has to be a partial function \emph{onto} $L_1$, so that $p^\dagger$ is an entire relation on which there can be $R$-relations.

Then we take an $R$-relational presheaf $\pi : E \to L_0$ over the initial tree $L_0$;
to keep describing the $(2, 2, 2)$-scenario as an example, let us take $E$ on $\S_{AB}$ in \eqref{eq:multipartite}.
Now, finally, we can provide interpretations $\Scott{\varphi}$ for sentences $\varphi$ of the language above by first assigning subsets $\Scott{p}$ to propositional letters $p$ and by then extending $\Scott{-}$ recursively.
We use the classical clauses for the Boolean connectives.

The new clauses of recursion for $\Scott{-}$ concern $\nec{e} \varphi$ and $P(\varphi \mid e) \gtreqqless r$.
Since $e$ may not lie in the initial tree $L_0$, let $p : L_0 \to L_1$ be the (perhaps partial) function onto the tree $L_1$ in which $e$ lies.
Then we express the ideas
\begin{itemize}
\item
$\nec{a} \varphi$ means that $\varphi$ will be the case when Alice chooses $a$, regardless of which of $b$ and $b'$ Bob may choose;
\item
$P(\varphi \mid a) \gtreqqless r$ means that the probability with which $\varphi$ will be the case when Alice chooses $a$ is greater than (equal to, or less than) $r$, regardless of which of $b$ and $b'$ Bob may choose;
\end{itemize}
with the following clauses:
\begin{itemize}
\item
$s \in \Scott{\nec{e} \varphi}$ iff $\supp{d^{E(e')}_s} \subseteq \Scott{\varphi}$ for all $e' = (\pi(s), x) \in p^{-1}(e)$;
\item
$s \in \Scott{P(\varphi \mid e) \gtreqqless r}$ iff $\sum_{t \in \S(e')(s) \cap \Scott{\varphi}} d^{E(e')}_s(t) \gtreqqless r$ for all $e' = (\pi(s), x) \in p^{-1}(e)$.
\end{itemize}
It is worth noting that $s$ may fail to be in any of $\Scott{P(\varphi \mid e) > r}$, $\Scott{\cdots = r}$ and $\Scott{\cdots < r}$, when no-signalling fails (this is why the three sentences cannot define each other).
On the other hand, $s \in \Scott{P(\varphi \mid a) = r}$ implies that the model satisfies no-signalling regarding $a$.

Given this semantics, the following axioms and rules are sound (we omit ones regarding probability modalities;
a complete axiomatization is an open problem):
\begin{itemize}
\item
Classical propositional logic.
\item
Standard axioms and rules for every $\nec{e}$:
\begin{align*}
\begin{array}{r@{}l}
\varphi & {} \vdash \psi \\ \hline
\nec{e} \varphi & {} \vdash \nec{e} \psi
\end{array} ,
&&
\vdash \nec{e} \top ,
&&
\nec{e} \varphi \wedge \nec{e} \psi \vdash \nec{e}(\varphi \wedge \psi) .
\end{align*}
\item
Moreover, whenever $e_0$, $e_1$ are such that $p(e_1) = e_0$ for one of the fibrations $p$,
\begin{align}
\label{eq:axiom.labels}
\nec{e_0} \varphi \vdash \nec{e_1} \varphi .
\end{align}
\item
In addition, because $R$-relations are on entire relations and distributions have nonempty supports, the semantics validates
\begin{align*}
\pos{e_0} \pos{e_1} \top \vdash \nec{e_0} \pos{e_1} \top .
\end{align*}
\end{itemize}

\subsection{Dynamic-Logical Characterization of Contextuality}\label{sec:characterization}

The dynamic logic and its semantics introduced in \autoref{sec:dynamic.logic} can provide characterization for contextuality, taking advantage of the kind of semantic structures we studied in \autoref{sec:hidden.variable}.

Recall the characterization of contextuality in \autoref{thm:global.section}.
That is, a model $E$ fails to be contextual iff obtained from a deterministic hidden-variable model $H$ by forgetting the middle stage---%
that is, iff consistent with the possibility that, at the middle stage, i.e., after $i$ and before $ab$ as in \eqref{eq:hidden.variable}, the states are deterministic instruction sets.
So, using labels and propositional letters for measurements and outcomes (as mentioned above in \autoref{sec:dynamic.logic}), let us introduce the sentence $\Det$ expressing determinacy:
\begin{align*}
\Det(a) & := \nec{a}(a = 0) \vee \nec{a}(a = 1) , &
\Det & := \Det(a) \wedge \cdots \wedge \Det(b') .
\end{align*}
Then the description of the Popescu-Rohrlich box \cite{pop94} of $(2, 2, 2)$-type,
\begin{align*}
\Delta_\text{PR} :=
\pos{ab} \top \wedge \cdots \wedge \pos{a'b'} \top 
& {} \wedge \nec{ab}(a = 0 \leftrightarrow b = 0) \wedge \nec{ab'}(a = 0 \leftrightarrow b' = 0) \\
& {} \wedge \nec{a'b}(a' = 0 \leftrightarrow b = 0) \wedge \nec{a'b'}(a' = 0 \oplus b' = 0) ,
\end{align*}
entails $\lnot \Det$, using the axioms and rules mentioned above, including the suitable ones of the form \eqref{eq:axiom.labels} such as $\nec{a} \varphi \vdash \nec{ab} \varphi$.
Also, a (partial) description $\Delta_\text{Hardy}$ of the Hardy model \cite{har93},
\begin{align*}
\Delta_\text{Hardy} :=
\pos{i} \pos{ab} \top \wedge \cdots \wedge \pos{i} \pos{a'b'} \top 
& {} \wedge \pos{i} \pos{ab}(a = 0 \wedge b = 0) \wedge \nec{i} \nec{ab'}(a = 1 \vee b' = 1) \\
& {} \wedge \nec{i} \nec{a'b}(a' = 1 \vee b = 1) \wedge \nec{i} \nec{a'b'}(a' = 0 \vee b' = 0)
\end{align*}
(note that this description involves label $i$), entails $\lnot \nec{i} \Det$ using the same axioms.

We can generalize these examples.
The upshot, roughly put, will be as follows.
\begin{itemize}
\item
The sentence $\lnot \nec{i} \Det$ characterizes contexuality.
\item
In addition, the sentence $\nec{i} \lnot \Det$ characterizes ``strong contextuality'' (see \cite[\Sect 6]{abr11} for definition).
\end{itemize}
To put this more precisely and to prove it, we need some notation and definitions.
First, fix a set $M$ of measurements, along with a family of measurement contexts (i.e., jointly performable sets of measurements), and for each measurement $a$ a set of outcomes $O_a$---%
we assume $\M$ and $O_a$ to be finite.
Then, in the vocabulary for $\M$ and $O_a$, write
\begin{itemize}
\item
$\Lambda$ for the set of axioms of the form \eqref{eq:axiom.labels}, for any pair of measurement contexts $e_0, e_1 \subseteq M$ such that $e_0 \subseteq e_1$;
and
\item
$\Det$ for the sentence $\bigwedge_{a \in \M} \Det(a)$, where $\Det(a) := \bigvee_{k \in O_a} \nec{a}(a = k)$.
\end{itemize}
Moreover,
\begin{itemize}
\item
By a ``legal'' sentence, let us mean a sentence of the form either $\nec{i} \nec{e} \varphi$, $\pos{i} \pos{e} \varphi$, or $P(\varphi \mid e \cmp i) \gtreqqless r$ in which $e$ is a jointly performable set of measurements and $\varphi$ is a Boolean compound of propositional letters referring to no measurements other than those in $e$.
\end{itemize}
It is clear that legal sentences can be used to describe empirical models (again, in the sense of \cite{abr11}).
More precisely, let $E$ be any empirical model involving $\RRp$- (or $\BB$-) distributions;
i.e., for each measurement context $e$, $E_e$ is a $\RRp$- (or $\BB$-) distribution on the set $\prod_{a \in e} O_a$.
Then
\begin{itemize}
\item
legal $\nec{i} \nec{e} \varphi$ describes $E$ iff $\varphi$ holds of every support of the distribution $E_e$;
\item
legal $\pos{i} \pos{e} \varphi$ describes $E$ iff $\varphi$ holds of some support of the distribution $E_e$;
\item
legal $P(\varphi \mid e \cmp i) \gtreqqless r$ describes $E$ iff $\sum_{f \in \supp{E_e} \text{ and } \varphi \text{ holds of } f} E_e(f) \gtreqqless r$.
\end{itemize}
%
%

Then we finally have

\begin{theorem}\label{thm:characterization}
Let $\Delta$ be a set of legal sentences that contains $\pos{i} \pos{e} \top$ for every maximal measurement context $e$.
Then the following are equivalent.
\begin{enumerate}
\item\label{item:characterization.contextual}
Every empirical model $E$ that $\Delta$ describes is contextual.
\item\label{item:characterization.not.nec.det}
Every stochastic relational presheaf model that satisfies no-signalling and validates $\Lambda$ validates \textup{$\Delta \vdash \lnot \nec{i} \Det$}.
\end{enumerate}
Moreover, the following are equivalent.
\begin{enumerate}
\addtocounter{enumi}{2}
\item\label{item:characterization.strong.contextual}
Every empirical model $E$ that $\Delta$ describes is strongly contextual.
\item\label{item:characterization.nec.not.det}
Every stochastic relational presheaf model that satisfies no-signalling and validates $\Lambda$ validates \textup{$\Delta \vdash \nec{i} \lnot \Det$}.
\end{enumerate}
\end{theorem}

\begin{proof}
Suppose \eqref{item:characterization.contextual} fails;
that is, there is an empirical model $E$ that $\Delta$ describes but that is not contextual, so that $E$ has a global section $E_\M$.
Then construct a stochastic relational presheaf model $\pi : H \to L_0$ as follows.
First build a tree $L_0$ with the edge $i : x \to y$ followed by the edges $e : y \to z_e$ for all the maximal measurement contexts $e$.
Then, for each non-maximal measurement context $e$, build a tree $L_e$ with $i$ followed by $e : y \to z_e$.
Between $L_0$ and $L_e$, we take a (typically partial) function $p_e : L_0 \to L_e$ that maps $i$ to $i$ and any $e'$ such that $e \subseteq e'$ to $e$;
and, whenever $e_0 \subseteq e_1$, we take a function $p_{e_1,e_0} : L_{e_1} \to L_{e_0}$ that maps $i$ to $i$ and $e_1$ to $e_0$.
Now let $H(x) = \{ s \}$;
$H(y) = \prod_{a \in \M} O_a$;
$H(z_e) = \prod_{a \in e} O_a$;
and, moreover, writing $|H(i)|$ and $|H(e)|$ for the underlying relations of the $\RRp$- (or $\BB$-) relations $H(i)$ and $H(e)$,
\begin{itemize}
\item
Let $|H(i)|(s) = \prod_{a \in \M} O_a$ and $d^{H(i)}_s(f) = E_\M(f)$ for each $f \in \prod_{a \in \M} O_a$.
\item
For each $f \in \prod_{a \in \M} O_a$ and maximal context $e$, let $|H(e)|(f) = \{ \rest{f}{e} \}$ (so that $d^{H(e)}_f$ is trivially deterministic).
\item
For each $a \in M$ and $k \in O_a$, let $\Scott{a = k} = \sum_{e \text{ is a maximal context and } a \in e} \{\, g \in \prod_{a \in e} O_a \mid g(a) = k \,\}$.
\end{itemize}
Then it is straightforward to check that $H$ satisfies no-signalling, that $H$ validates $\Lambda$, and that $s \in \Scott{\psi}$ for each $\psi \in \Delta$ since $\psi$ describes $E$.
Yet, since each $f \in \prod_{a \in \M} O_a$ satisfies $f \in \bigcap_{a \in \M} \Scott{\nec{a}(a = f(a))} \subseteq \Scott{\Det}$, we have $s \in \Scott{\nec{i} \Det}$, so $s \notin \Scott{\lnot \nec{i} \Det}$.
Therefore $H$ does not validate $\Delta \vdash \lnot \nec{i} \Det$.
Thus \eqref{item:characterization.not.nec.det} fails.

On the other hand, assuming \eqref{item:characterization.not.nec.det} fails, let $\pi : H \to L_0$ be a stochastic relational presheaf model that validates $\Lambda$ but that has some $s \in \bigcap_{\psi \in \Delta} \Scott{\psi}$ with $s \notin \Scott{\lnot \nec{i} \Det}$.
Since $s \in \Scott{\pos{i} \pos{e} \top}$ for each maximal context $e$, $L_0$ has edges $i : \pi(s) \to y$ and $e : y \to z_e$ for all the maximal contexts $e$.
Then define a $\RRp$- (or $\BB$-) distribution $E_\M$ on $\prod_{a \in \M} O_a$ so that, for each $f \in \prod_{a \in \M} O_a$,
\begin{align*}
E_\M(f) = \sum_{t \in |H(i)|(s) \cap \bigcap_{a \in \M} \Scott{\nec{a}(a = f(a))}} d^{H(i)}_s(t) .
\end{align*}
Also, for each maximal context $e$, define a distribution $E_e$ on $\prod_{a \in e} O_a$ so that, for each $f \in \prod_{a \in e} O_a$,
\begin{align*}
E_e(f) = \sum_{u \in |H(e)| \cmp |H(i)|(s) \cap \bigcap_{a \in e} \Scott{a = f(a)}} d^{H(e) \cmp H(i)}_s(u) .
\end{align*}
Then it is easy to check that the family $E = \{ E_e \}_e$ satisfies the compatibility condition (i.e., no-signalling) and hence is an empirical model, that $\Delta$ describes $E$, and that $E_\M$ is a global section for $E$.
Thus \eqref{item:characterization.contextual} fails.

Suppose \eqref{item:characterization.strong.contextual} fails;
that is, there is an empirical model $E$ that $\Delta$ describes but that is not strongly contextual, so that there is a function $f : \prod_{a \in \M} O_a$ such that $\rest{f}{e} \in \supp{E_e}$ for every context $e$.
Then construct a stochastic relational presheaf model $\pi : H \to L_0$ as follows.
First build trees $L_0$ and $L_e$ as in the first paragraph of this proof.
Now let $H(x) = \{ s \}$;
$H(y) = \{ f, t \}$;
$H(z_e) = \prod_{a \in e} O_a$ for each maximal context $e$;
and, moreover, define $H(i)$ and each $H(e)$ as the following $\RRp$- (or $\BB$-) relations (we lay out how to define $\RRp$-relations, since we can use their supports to define $\BB$-relations).
\begin{itemize}
\item
$|H(i)|(s) = \{ f, t \}$, and $d^{H(i)}_s(f) = \min \{\, E_e(\rest{f}{e}) \mid e$ is a maximal context$\,\}$, so that $d^{H(i)}_s(t) = 1 - d^{H(i)}_s(f)$.
\item
$|H(e)|(f) = \{ \rest{f}{e} \}$ (so that $d^{H(e)}_f$ is trivially deterministic).
\item
$|H(e)|(t) = \prod_{a \in e} O_a$.
If $d^{H(i)}_s(t) = 0$, then $d^{H(e)}_t(g) = 1$ for all $g \in \prod_{a \in e} O_a$.
Otherwise
\begin{align*}
d^{H(e)}_t(g) =
\begin{cases}
\, \dfrac{E_e(g) - d^{H(i)}_s(f)}{d^{H(i)}_s(t)} & \text{if } g = \rest{f}{e} , \\
\, \dfrac{E_e(g)}{d^{H(i)}_s(t)} & \text{otherwise.}
\end{cases}
\end{align*}
\end{itemize}
Lastly, 
%
%
for each $a \in M$ and $k \in O_a$, let $\Scott{a = k} = \sum_{e \text{ is a maximal context and } a \in e} \{\, g \in \prod_{a \in e} O_a \mid g(a) = k \,\}$.
%
%
Then it is straightforward to check that $H$ satisfies no-signalling, that $H$ validates $\Lambda$, and that $s \in \Scott{\psi}$ for each $\psi \in \Delta$ since $\psi$ describes $E$.
Yet $f \in \bigcap_{a \in \M} \Scott{\nec{a}(a = f(a))} \subseteq \Scott{\Det}$ implies $s \notin \Scott{\nec{i} \lnot \Det}$.
Therefore $H$ does not validate $\Delta \vdash \nec{i} \lnot \Det$.
Thus \eqref{item:characterization.nec.not.det} fails.

On the other hand, assuming \eqref{item:characterization.nec.not.det} fails, let $\pi : H \to L_0$ be a stochastic relational presheaf model that validates $\Lambda$ but that has some $s \in \bigcap_{\psi \in \Delta} \Scott{\psi}$ with $s \notin \Scott{\nec{i} \lnot \Det}$.
As before, since $s \in \Scott{\pos{i} \pos{e} \top}$ for each maximal context $e$, $L_0$ has edges $i : \pi(s) \to y$ and $e : y \to z_e$ for all the maximal contexts $e$.
Then, as before, for each maximal context $e$, define a distribution $E_e$ on $\prod_{a \in e} O_a$ so that, for each $f \in \prod_{a \in e} O_a$,
\begin{align*}
E_e(f) = \sum_{u \in |H(e)| \cmp |H(i)|(s) \cap \bigcap_{a \in e} \Scott{a = f(a)}} d^{H(e) \cmp H(i)}_s(u) .
\end{align*}
Again it is easy to check that the family $E = \{ E_e \}_e$ is an empirical model that $\Delta$ describes.
Now note that, since $s \notin \Scott{\nec{i} \lnot \Det}$, that is, since $s \in \Scott{\pos{i} \Det}$, some $t \in \supp{d^{H(i)}_s}$ lies in $\Scott{\Det}$ and therefore we have a function $f \in \prod_{a \in \M} O_a$ with which each $a \in \M$ has $t \in \Scott{\nec{a}(a = f(a))}$;
this implies that, for each maximal context $e$, there is $u \in |H(e)| \cmp |H(i)|(s) \cap \bigcap_{a \in e} \Scott{a = f(a)}$ such that $d^{H(e) \cmp H(i)}_s(u) \geqslant d^{H(i)}_s(t) \cdot d^{H(e)}_t(u) > 0$, which means that $\rest{f}{e} \in \supp{E_e}$.
Hence $E$ is not strongly contextual.
Thus \eqref{item:characterization.strong.contextual} fails.
\end{proof}

\section{Conclusion}

In this paper, we have integrated the three frameworks mentioned in the Introduction for capturing non-deterministic processes, \eqref{item:presheaf}--\eqref{item:stoch}, by introducing the category $R\RRel$ of $R$-relations and taking $R$-relational presheaves---%
functors from trees to $R\RRel$.
The resulting structure captures stochastic dynamics with a good enough expressive power, as demonstrated by the fact that it provides a labelled transitional formulation for the sheaf-theoretic approach of Abramsky and Brandenburger \cite{abr11} to non-locality and contextuality, and moreover yielding dynamic logic with a modal-logical characterization of contextuality.
(In fact, our formalism is partially equivalent to the sheaf-theoretic approach, extending the equivalence between presheaves and fibrations.)
Whereas the sheaf-theoretic approach can take advantage of methods of cohomology to calculate conditions for contextuality (see \cite{abr12}), our approach on the other hand has a certain flexibility in the base trees of measurement labels, so that it can readily express contextuality in not just one round of measurements but within a sequence or protocol of measurements.
Thus our approach is expected to complement the sheaf-theoretic approach and extend it to various applications.
Needless to say, applications to other kinds of stochastic dynamics can be expected as well.

\nocite{*}
\bibliographystyle{eptcs}
\bibliography{qpl.2014}

\end{document}